%% file: ridgeLeverageScores.tex
\documentclass[11pt]{article}
\usepackage{fullpage}
\usepackage{amsmath,amsfonts,amsthm,amssymb}
\usepackage{url}

\usepackage{color}
\usepackage[usenames,dvipsnames,svgnames,table]{xcolor}
\usepackage[colorlinks=true, linkcolor=black, urlcolor=blue, citecolor=gray]{hyperref}

\usepackage{algorithm}
\usepackage{algpseudocode}
\usepackage{float}
\makeatletter
\newenvironment{breakablealgorithm}
  {
   \begin{flushleft}
     \refstepcounter{algorithm}
     \hrule height.8pt depth0pt \kern2pt
     \renewcommand{\caption}[2][\relax]{
       {\raggedright\textbf{\ALG@name~\thealgorithm} ##2\par}%
       \ifx\relax##1\relax 
         \addcontentsline{loa}{algorithm}{\protect\numberline{\thealgorithm}##2}%
       \else 
         \addcontentsline{loa}{algorithm}{\protect\numberline{\thealgorithm}##1}%
       \fi
       \kern2pt\hrule\kern2pt
     }
  }{
     \kern2pt\hrule\relax
   \end{flushleft}
  }
\makeatother

\DeclareMathOperator*{\E}{\mathbb{E}}
\let\Pr\relax
\DeclareMathOperator*{\Pr}{\mathbb{P}}
\DeclareMathOperator{\tr}{tr}
\DeclareMathOperator{\fd}{FreqDirUpdate}
\DeclareMathOperator{\re}{ApproximateRidgeScores}

\DeclareMathOperator{\colspan}{span}
\newcommand{\R}{\mathbb{R}}

\DeclareMathOperator*{\argmin}{arg\,min}

\newcommand{\poly}{\mathop\mathrm{poly}}

\DeclareMathOperator{\nnz}{nnz}

\newcommand{\eqdef}{\mathbin{\stackrel{\rm def}{=}}}
\newcommand{\norm}[1]{\|#1\|}

\newcommand{\bs}[1]{\boldsymbol{#1}}
\newcommand{\bv}[1]{\mathbf{#1}}

\usepackage{color}

\newcommand{\algoname}[1]{\textnormal{\textsc{#1}}}
\makeatletter

\newtheorem*{rep@theorem}{\rep@title}
\newcommand{\newreptheorem}[2]{%
\newenvironment{rep#1}[1]{%
 \def\rep@title{#2 \ref{##1}}%
 \begin{rep@theorem}}%
 {\end{rep@theorem}}}
\makeatother
\newtheorem{theorem}{Theorem}
\newreptheorem{theorem}{Theorem}

\newtheorem{lemma}[theorem]{Lemma}
\newtheorem*{lemma*}{Lemma}
\newreptheorem{lemma}{Lemma}

\newtheorem{definition}[theorem]{Definition}

\usepackage[margin=1in]{geometry}

\input{bib_macros}

\title{Input Sparsity Time Low-Rank Approximation \\ via Ridge Leverage Score Sampling 
}
\author{Michael B. Cohen \hspace{4em} Cameron Musco  \hspace{4em} Christopher Musco\thanks{Part of this work was completed while the author interned at Yahoo Labs, NYC.}
\vspace{1em}
\\ \textit{Massachusetts Institute of Technology, EECS}
\\ \textit{Cambridge, MA 02139, USA}
\\ Email: \{micohen,cnmusco,cpmusco\}@mit.edu}

\begin{document}
\maketitle

%
%

\begin{abstract}
We present a new algorithm for finding a near optimal low-rank approximation of a matrix $\bv{A}$ in $O(\nnz(\bv{A}))$ time. Our method is based on a recursive sampling scheme for computing a representative subset of $\bv{A}$'s columns, which is then used to find a low-rank approximation.

This approach differs substantially from prior $O(\nnz(\bv{A}))$ time algorithms, which are all based on fast Johnson-Lindenstrauss random projections. It matches the guarantees of these methods while offering a number of advantages.

Not only are sampling algorithms faster for sparse and structured data, but they can also be applied in settings where random projections cannot. For example, we give new single-pass streaming algorithms for the column subset selection and projection-cost preserving sample problems. Our method has also been used to give the fastest algorithms for provably approximating kernel matrices \cite{kernelStuff}.
\end{abstract}

\thispagestyle{empty}
\clearpage
\setcounter{page}{1}


\section{Introduction}
Low-rank approximation is a fundamental task in statistics, machine learning, and computational science. The goal is to find a rank $k$ matrix that is as close as possible to an arbitrary input matrix $\bv{A} \in \R^{n\times d}$, with distance typically measured using the spectral or Frobenius norms.

Traditionally, the problem is solved using the singular value decomposition (SVD), which takes $O(nd^2)$ time to compute. 
This high cost can be reduced using iterative algorithms like the power method or Krylov methods, which require just $O(\nnz(\bv{A})\cdot k)$ time per iteration, where $\nnz(\bv{A})$ is the number of non-zero entries in $\bv{A}$\footnote{The number of iterations depends on the accuracy $\epsilon$ and/or spectral gap conditions. See \cite{blanczos} for an overview.}.

More recently, the cost of low-rank approximation has been reduced even further using sketching methods based on Johnson-Lindenstrauss random projection \cite{sarlos2006improved}. Remarkably, so-called ``sparse random projections''  \cite{clarkson2013low,meng2013,osnap, bourgain,Cohen:2016} give algorithms that run in time\footnote{$\tilde{O}(\cdot)$ hides logarithmic factors, including a failure probability dependence.}:
\begin{align*}
O(\nnz(\bv{A})) + \tilde{O}\left(n\cdot \poly(k,\epsilon)\right).
\end{align*}
These methods output a low-rank approximation within a $(1+\epsilon)$ factor of optimal when error is measured using the Frobenius norm. They are typically referred to as running in ``input sparsity time'' since the $O(\nnz(\bv{A}))$ term is considered to dominate the runtime. 

Input sparsity time algorithms for low-rank approximation were an important theoretical achievement and have also been influential in practice. Implementations are now available in a variety of languages and machine learning libraries \cite{matlabrandsvd,cpprandsvd,libskylark,scikit-learn,scalarandsvd,martinsonImp}. 

\subsection{Our Contributions}
We give an entirely different approach to obtaining input sparsity time algorithms for low-rank approximation. Random projection methods are based on multiplying $\bv{A}$ by a sparse random matrix $\bs{\Pi} \in \R^{d\times poly(k,\epsilon)}$ to form a smaller matrix $\bv{A}\bs{\Pi}$ that contains enough information about $\bv{A}$ to compute a near optimal low-rank approximation. Our techniques on the other hand are based on \emph{sampling} $\tilde O(k/\epsilon)$ columns from $\bv{A}$, and computing a low-rank approximation using this sample.

Sampling itself is simple and extremely efficient. However, to obtain a good approximation to $\bv{A}$, columns must be sampled with non-uniform probabilities, carefully chosen to reflect their relative importance. It is known that variations on the standard statistical leverage scores give probabilities that are provably sufficient for low-rank approximation \cite{sarlos2006improved,betterSubspaceScores,kmeansPaper}.

Unfortunately, computing any of these previously studied ``low-rank leverage scores'' is \emph{as difficult as low-rank approximation itself}, so sampling did not yield fast algorithms\footnote{$\ell_2$ norm sampling does yield very fast algorithms \cite{Frieze:2004,Drineas:2006:FMC2,bhojanapalli2014tighter}, but cannot give relative error guarantees matching those of random projection or leverage score methods without additional assumptions on $\bv{A}$.}. 

We address this issue for the first time by introducing new importance sampling probabilities which can be approximated efficiently using a simple recursive algorithm. In particular, we adapt the so-called \emph{ridge leverage scores} to low-rank matrix approximation. These scores have been used as sampling probabilities in the context of linear regression and spectral approximation \cite{pengV2,kapralov2014single,alaoui2014fast} but never for low-rank approximation. 

By showing that ridge leverage scores display a unique monotonicity property under perturbations to $\bv{A}$, we are able to prove that, unlike any prior low-rank leverage scores, they can be approximated using a relatively large \emph{uniform subsample} of $\bv{A}$'s columns. While too large to use directly, the size of this subsample can be reduced recursively to give an overall fast algorithm. This approach resembles work on recursive methods for computing standard leverage scores, which were recently used to give the first $O(\nnz(\bv{A}))$ time sampling algorithms for linear regression \cite{pengV2, cohen2015uniform}.

Our main algorithmic result, which nearly matches the state-of-the-art in \cite{osnap} follows:

\begin{theorem}\label{nnza_intro} For any $\theta \in (0,1]$, there exists a recursive column sampling algorithm that, in time $O\left (\theta^{-1}\nnz(\bv{A})\right ) + \tilde O \left (\frac{n^{1+\theta}k^{2}}{\epsilon^4}\right)$, returns $\bv{Z} \in \mathbb{R}^{n\times k}$ satisfying:
\begin{align}\label{low_rank_gurantee_orig}
\norm{\bv{A}-\bv{ZZ}^T \bv{A}}_F^2 \le (1+\epsilon) \norm{\bv A - \bv{A}_k}_F^2.
\end{align}
\end{theorem}
Here $\bv{A}_k$ is the best rank $k$ approximation to $\bv{A}$. 
To prove Theorem \ref{nnza_intro}, we show how to compute a sampling matrix $\bv{S}$ such that $\bv{AS} \in \R^{n\times \tilde{O}(k/\epsilon^2)}$ satisfies a \emph{projection-cost preservation} guarantee (formalized in Section \ref{sec:tech_intro}). This property ensures that it is possible to extract a near optimal low-rank approximation from the sample. It also allows $\bv{A}\bv{S}$ to be used to approximately solve a broad class of constrained low-rank approximation problems, including $k$-means clustering \cite{kmeansPaper}.

With a slightly smaller sample, we also prove that $\bv{AS}$ satisfies a standard $(1+\epsilon)$ error \emph{column subset selection} guarantee. Ridge leverage score sampling is the first algorithm, efficient or otherwise, that obtains both of these important approximation goals simultaneously.

\subsection{Why Sampling?}
Besides the obvious goal of obtaining alternative state-of-the-art algorithms for low-rank approximation, we are interested in sampling methods for a few specific reasons:

\medskip
\noindent \textbf{Sampling maintains matrix sparsity and structure.}

\noindent Without additional assumptions on $\bv{A}$, our recursive sampling algorithms essentially match random projection methods. However, they have the potential to run faster for sparse or structured data. Random projection linearly combines \emph{all} columns in $\bv{A}$ to form $\bv{A}\bv{\Pi} \in \R^{n \times poly(k,\epsilon)}$, so this sketched matrix is usually dense and unstructured. On the other hand, $\bv{A}\bv{S}$ will remain sparse or structured if $\bv{A}$ is sparse or structured, in which case it can be faster to post-process.
Potential gains are especially important when the $\tilde{O}(n\cdot \poly(k,\epsilon))$ runtime term is not dominated by $O(\nnz(\bv{A}))$.

We note that the ability to maintain sparsity and structure motivated similar work on recursive sampling algorithms for fast linear regression \cite{cohen2015uniform}. While these techniques only match random projection for general matrices, they have been important ingredients in designing faster algorithms for highly structured Laplacian and SDD matrices \cite{parallelSDD1, parallelSDD2,otherParallel}. We hope our sampling methods for low-rank approximation will provide a foundation for similar contributions.

\medskip
\noindent \textbf{Sampling techniques lead to natural streaming algorithms.} 

\noindent 
In data analysis, sampling itself is often \emph{the primary goal}. The idea is to select a subset of columns from $\bv{A}$ whose span contains a good low-rank approximation for the matrix and hence represents important or influential features \cite{paschou2007pca,mahoney2009cur}.

Computing this subset in a streaming setting is of both theoretical and practical interest \cite{Strauch2014}.
Unfortunately, while  random projection methods adapt naturally to data streams \cite{clarkson:2009}, importance sampling is more difficult. The leverage score of one column depends on every other column,  including those that have not yet appeared in the stream. While random projections can be used to approximate leverage scores, this approach inherently requires two passes over the data. 

Fortunately, the same techniques used in our recursive algorithms apply naturally in the streaming setting. We can compute coarse approximations to the ridge leverage scores using just a small number of columns and refine these approximations as the stream is revealed. By rejection sampling columns as the probabilities are adjusted, we obtain the first space efficient single-pass streaming algorithms for both column subset selection and projection-cost preserving sampling (Section \ref{sec:streaming}).

\medskip
\noindent \textbf{Sampling offers additional flexibility for non-standard matrices.}

\noindent
In recent follow up work, the techniques in this paper are adapted to give the most efficient, provably accurate algorithms for kernel matrix approximation \cite{kernelStuff}. In nearly all settings this well-studied problem cannot be solved efficiently by random projection methods.

The goal in kernel approximation is to replace an $n \times n$ positive semidefinite kernel matrix $\bv{K}$ with a low-rank approximation that takes less space to represent \cite{ams01,originalNystrom,fine2002efficient,mahoneyDrineasNystrom,rahimi2007random,belabassNystrom2,belabassNystrom1,bach2012sharp,gittens2013revisiting,ipsen,chengtaodpp}. 
However, unlike in the standard low-rank approximation problem, $\bv{K}$ is not represented explicitly. Its entries can only be accessed by evaluating a ``kernel function'' between each pair of the $n$ points in a data set. 

Sketching $\bv{K}$ using random projection requires computing the full matrix first, using $\Theta(n^2)$ kernel evaluations. On the other hand, with recursive ridge leverage score sampling this is not necessary -- it is possible to compute entries of $\bv{K}$ `on the fly', only when they are required to compute ridge scores with respect to a subsample. \cite{kernelStuff}  shows that this technique 
gives the first provable algorithms for approximating kernel matrices that only require time \emph{linear in $n$}. In other words, the methods only evaluate a tiny fraction of the dot products required to build $\bv{K}$. Notably they do not require any coherence or regularity assumptions to achieve this runtime.

Aside from kernel approximation, we note that in \cite{bhojanapalli2014tighter} the authors present a low-rank approximation algorithm based on elementwise sampling that they show can be applied to the product of two matrices without ever forming this product explicitly. This result again highlights the flexibility of sampling-based methods for non-standard matrices. Without access to an efficiently computable leverage score distribution for elementwise sampling,  \cite{bhojanapalli2014tighter} applies an approximation based on $\ell_2$ sampling. This approximation only performs well under additional assumptions on $\bv{A}$ and an interesting open question is to see if our techniques can be adapted to their framework in order to eliminate such assumptions.


\subsection{Techniques and Paper Layout}

\medskip
\noindent \textbf{Sampling Bounds (Sections \ref{sec:tech_intro}, \ref{sec:prelims})}:
In Section \ref{sec:tech_intro} we review technical background and introduce ridge leverage scores. In Section \ref{sec:prelims} we prove that sampling by ridge leverage scores gives solutions to the projection-cost preserving sketch and column subset selection problems. These sections do not address algorithmic considerations. 

While the proofs are technical, we reduce both problems to a simple ``additive-multiplicative spectral guarantee,'' which resembles the ubiquitous subspace embedding guarantee \cite{sarlos2006improved}. This approach greatly simplifies prior work on low-rank approximation bounds for sampling methods \cite{kmeansPaper} and we hope that it will prove generally useful in studying future sketching methods.

\medskip
\noindent \textbf{Ridge Leverage Score Monotonicity (Section \ref{sec:mono})}: In Section \ref{sec:mono} we prove a basic theorem regarding the stability of ridge leverage scores. Specifically, we show that the ridge leverage score of a column cannot increase if another column is added to the matrix. This fact, which do not hold for any prior ``low-rank leverage scores'', is essential in proving the correctness of our recursive sampling procedure and streaming algorithms. 

\medskip
\noindent \textbf{Recursive Sampling Algorithm (Section \ref{sec:nnza})}:
In Section \ref{sec:nnza}, we describe and prove the correctness of our main sampling algorithm. We show how a careful implementation of the algorithm gives $O(\nnz(\bv{A}))$ running time for computing ridge leverage scores, and accordingly for solving the low-rank approximation problem.

\medskip
\noindent \textbf{Application to Streaming (Section \ref{sec:streaming})}:
We conclude with an application of our results to low-rank sampling algorithms for single-pass column streams that are only possible thanks to the stability result of Section \ref{sec:mono}.

\section{Technical Background}
\label{sec:tech_intro}

\subsection{Low-rank Approximation}

Using the singular value decomposition (SVD), any rank $r$ matrix $\bv{A} \in \R^{n\times d}$ can be factored as $\bv{A} = \bv{U}\bs{\Sigma}\bv{V}^T$. $\bv{U} \in \R^{n\times r}$ and $\bv{V} \in \R^{d\times r}$ are orthonormal matrices whose columns are the left and right singular vectors of $\bv{A}$. $\bv{\Sigma}$ is a diagonal matrix containing $\bv{A}$'s singular values $\sigma_1 \geq \sigma_2 \geq \ldots \geq \sigma_r > 0$ in decreasing order from top left to bottom right. When quality is measured with respect to the Frobenius norm, the best low-rank approximation for $\bv{A}$ is given by $\bv{A}_k = \bv{U}_k\bs{\Sigma}_k\bv{V}_k^T$ where $\bv{U}_k \in \mathbb{R}^{n\times k}$, $\bv{V}_k \in \mathbb{R}^{d\times k}$, and $\bs{\Sigma}_k \in \mathbb{R}^{k \times k}$ contain the just the first $k$ components of $\bv{U}$, $\bv{V}$, and $\bs{\Sigma}$ respectively. In other words,
\begin{align*}
\bv{A}_k = \argmin_{\bv{B}: rank(\bv{B}) \leq k} \|\bv{A} - \bv{B}\|_F.
\end{align*}

Since $\bv{U}$ has orthonormal columns, we can rewrite $\bv{A}_k = \bv{U}_k\bv{U}_k^T\bv{A}$. That is, the best rank $k$ approximation can be found by projecting $\bv{A}$ onto the span of its top $k$ singular vectors. Throughout, we will use the shorthand $\bv{A}_{\setminus k}$ to denote the residual $\bv{A} - \bv{A}_k$. $\bv{U}_{\setminus k} \in \mathbb{R}^{n\times r-k}$, $\bv{V}_{\setminus k} \in \mathbb{R}^{d \times r-k}$, and $\bv{\Sigma}_{\setminus k} \in \mathbb{R}^{r-k\times r-k}$ denote $\bv{U}$, $\bv{V}$, and $\bv{\Sigma}$ restricted to just their last $k$ components.

When solving the low-rank approximation problem approximately, our goal is to find an orthonormal span $\bv{Z} \in \R^{n\times k}$ satisfying $\|\bv{A} - \bv{ZZ}^T\bv{A}\|_F \leq (1+\epsilon)\|\bv{A} - \bv{U}_k\bv{U}_k^T\bv{A}\|_F$. 

\subsection{Sketching Algorithms}

Like many randomized linear algebra routines, our low-rank approximation algorithms are based on 
``linear sketching''. Sketching algorithms use a typically randomized procedure to compress $\bv{A} \in \R^{d\times n}$ into an approximation (or ``sketch'') $\bv{C} \in \R^{d'\times n}$ with many fewer columns ($d' \ll d$). Random projection algorithms construct $\bv{C}$ by forming $d'$ random linear combinations of the columns in $\bv{A}$. Random sampling algorithms construct $\bv{C}$ by selecting and possibly reweighting a $d'$ columns in $\bv{A}$.  

After compression, a post-processing routine, which is often deterministic, is used to solve the original linear algebra problem with just the information contained in $\bv{C}$. In our case, the post-processing step needs to extract an approximation to the span of $\bv{A}$'s top left singular vectors. If $\bv{C}$ is much smaller than $\bv{A}$, the cost of post-processing is typically considered a low-order term in comparison to the cost of computing the sketch to begin with.

When analyzing sketching algorithms it is common to separate the post-processing step from the dimensionality reduction step. Known post-processing routines give good approximate solutions to linear algebra problems under the condition that $\bv{C}$ satisfies certain approximation properties with respect to $\bv{A}$. The challenge then becomes proving that a specific dimensionality reduction algorithm produces a sketch satisfying these required guarantees.

\subsection{Sampling Guarantees for Low-rank Approximation}

For low-rank approximation, most algorithms aim for one of two standard approximation guarantees, which we describe below. Since we will be focusing on sampling methods, from now on we assume that $\bv{C}$ is a subset of $\bv{A}$'s columns. 

\begin{definition}[Rank $k$ Column Subset Selection]
\label{def:css}
For $d' < d$, a subset of $\bv{A}$'s columns $\bv{C}\in \R^{n\times d'}$ is a $(1+\epsilon)$ factor column subset selection if there exists a rank $k$ matrix $\bv{Q}\in \R^{d'\times d}$ with
\begin{align}
\label{eq:css}
\|\bv{A} - \bv{CQ}\|_F^2 \leq (1+\epsilon)\|\bv{A} - \bv{A}_k\|_F^2.
\end{align}
\end{definition}
In other words, the column span $\bv{C}$ contains a good rank $k$ approximation for $\bv{A}$. Algorithmically, we can recover this low-rank approximation via projection to the column subset  \cite{sarlos2006improved,clarkson2013low}.

Beyond sketching for low-rank approximation, the column subset selection guarantee is used as a metric in feature selection for high dimensional datasets \cite{paschou2007pca,mahoney2009cur}. With columns of $\bv{A}$ interpreted as features and rows as data points, \eqref{eq:css} ensures that we select $d'$ features that span the feature space nearly as well as the top $k$ principal components. 
The guarantee is also important in algorithms for CUR matrix decomposition \cite{Drineas:20063,mahoney2009cur,Boutsidis:2014} and Nystr\"{o}m approximation \cite{originalNystrom,mahoneyDrineasNystrom,belabassNystrom1,belabassNystrom2,gittens2013revisiting,ipsen,kernelStuff}.

In addition to Definition \ref{def:css}, we consider a stronger guarantee for  \emph{weighted} column selection, which has a broader range of algorithmic applications:

\begin{definition}[Rank $k$ Projection-Cost Preserving Sample]
\label{def:pcp}
For $d' < d$, a subset of rescaled columns $\bv{C}\in \R^{n\times d'}$ is a $(1+\epsilon)$ projection-cost preserving sample if, for all rank $k$ orthogonal projection matrices $\bv{X} \in \R^{n\times n}$,
\begin{align}
\label{eq:pcp}
(1-\epsilon)\|\bv{A} - \bv{X}\bv{A}\|^2_F \leq \|\bv{C} - \bv{X}\bv{C}\|^2_F\leq (1+\epsilon)\|\bv{A} - \bv{X}\bv{A}\|^2_F
\end{align}
\end{definition}

Definition \ref{def:pcp} is formalized in two recent papers \cite{feldman2013turning,kmeansPaper}, though it appears implicitly in prior work \cite{nphardkmeans,DBLP:journals/corr/abs-1110-2897}. It ensures that $\bv{C}$ approximates the cost of any rank $k$ column projection of $\bv{A}$. $\bv{C}$ can thus be used as a \emph{direct surrogate} of $\bv{A}$ to solve low-rank projection problems. Specifically, it's not hard to see that if we use a post-processing algorithm that sets $\bv{Z}$ equal to the top $k$ left singular vectors of $\bv{C}$, it will hold that $\|\bv{A} - \bv{ZZ}^T\bv{A}\|_F \leq (1+\epsilon)\|\bv{A} - \bv{U}_k\bv{U}_k^T\bv{A}\|_F$ \cite{kmeansPaper}. 
 
Definition \ref{def:pcp} also ensures that $\bv{C}$ can be used in approximately solving a variety of constrained low-rank approximation problems, including $k$-means clustering of $\bv{A}$'s rows (see \cite{kmeansPaper}).

\subsection{Leverage Scores}

It is well known that sketches satisfying Definitions \ref{def:css} and \ref{def:pcp} can be constructed via importance sampling routines which select columns using carefully chosen, non-uniform probabilities. Many of these probabilities
are modifications on traditional ``statistical leverage scores''.

The statistical leverage score of the $i^\text{th}$ column $\bv{a}_i$ of $\bv{A}$ is defined as\footnote{$+$ denotes the Moore-Penrose pseudoinverse of a matrix. When $\bv{AA}^T$ is full rank $(\bv{A}\bv{A}^T)^{+} = (\bv{A}\bv{A}^T)^{-1}$}: 
\begin{align}
\label{lev_definition}
\tau_i \eqdef \bv{a}_i^T (\bv{A}\bv{A}^T)^{+} \bv{a}_i.
\end{align}
$\tau_i$ measures how important $\bv{a}_i$ is in composing the range of $\bv{A}$. 
It is maximized at $1$ when $\bv{a}_i$ is linearly independent from $\bv{A}$'s other columns and decreases when many other columns approximately align with $\bv{a}_i$ or when $\|\bv{a}_i\|_2$ is small. 


Leverage scores are used in fast sketching algorithms for linear regression and matrix preconditioning
\cite{Drineas:2006,sarlos2006improved,cohen2015uniform}. They have also been applied to convex optimization \cite{cuttingPlane}, linear programming \cite{linearProg,inverseMain}, matrix completion \cite{wardCoherent}, multi-label classification \cite{bi2013efficient}, and graph sparsification, where they are known as \emph{effective resistances} \cite{Spielman:2008}.

\subsection{Existing Low-rank Leverage Scores}

For low-rank approximation problems, leverage scores need to be modified to only capture how important each column $\bv{a}_i$ is in composing the \emph{top few} singular directions of $\bv{A}$'s range.  

In particular, it is known that a sketch $\bv{C}$ satisfying Definition \ref{def:css} can be constructed by sampling $d' = O(k\log k + k/\epsilon)$ columns according to the so-called \emph{rank $k$ subspace scores} \cite{sarlos2006improved,betterSubspaceScores}:
\begin{align}
\label{eq:ss}
\text{$i^{\text{th}}$ rank $k$ subspace score:  \hspace{2em}}ss^{(k)}_i \eqdef \bv{a}_i^T (\bv{A}_k\bv{A}_k^T)^{+} \bv{a}_i.
\end{align}
These scores are exactly equivalent to standard leverage scores computed with respect to $\bv{A}_k$, an optimal low-rank approximation for $\bv{A}$.
The stronger projection-cost preservation guarantee of Definition \ref{def:pcp} can be achieved by sampling $O(k\log k/\epsilon^2)$ columns using a related, but somewhat more complex, leverage score modification \cite{kmeansPaper}.


\subsection{Ridge Leverage Scores}

Notably, prior low-rank leverage scores are defined in terms of $\bv{A}_k$, which is not always unique and regardless can be sensitive to matrix perturbations\footnote{It is often fine to use a near-optimal low-rank approximation in place of $\bv{A}_k$, but similar instability issues remain.}.
As a result, the scores can change drastically when $\bv{A}$ is modified slightly or when only partial information about the matrix is known. This largely limits the possibility of quickly approximating the scores with sampling algorithms, and motivates our adoption of a new leverage score for low-rank approximation.

Rather than use scores based on $\bv{A}_k$, we employ regularized scores called  \emph{ridge leverage scores}, which have been used for approximate kernel ridge regression \cite{alaoui2014fast} and in work on iteratively computing standard leverage scores \cite{pengV2,kapralov2014single}. We extend their applicability to low-rank approximation. 
For a given regularization parameter $\lambda$, define the $\lambda$-ridge leverage score as:
\begin{align}
\label{ridge_definition}
\tau^\lambda_i(\bv{A}) \eqdef \bv{a}_i^T (\bv{A A}^T + \lambda\bv{I})^+ \bv{a}_i.
\end{align}
We will always set $\lambda = \norm{\bv{A} - \bv{A}_k}_F^2/k$ and thus, for simplicity, use ``$i^\text{th}$ ridge leverage score'' to refer to $\bar \tau_i(\bv{A}) = \bv{a}_i^T \left (\bv{AA}^T + \frac{\norm{\bv{A} - \bv{A}_k}_F^2}{k} \bv{I} \right)^+ \bv{a}_i$. 

For prior low-rank leverage scores, $\bv{A}_k$ truncates the spectrum of $\bv{A}$, removing all but its top $k$ singular values. Regularization offers a smooth alternative: adding $\lambda\bv{I}$ to $\bv{AA}^T$ `washes out' small singular directions, causing them to be sampled with proportionately lower probability. 

This paper proves that regularization can not only replace truncation, but is more natural and stable.
In particular, while $\bar \tau_i$ depends on the \emph{value} of $\norm{\bv{A} - \bv{A}_k}_F^2$, it does not depend on a specific low-rank approximation. This is sufficient for stability since $\norm{\bv{A} - \bv{A}_k}_F^2$ changes predictably under matrix perturbations even when $\bv{A}_k$ itself does not.

Before showing our sampling guarantees for ridge leverage scores, we prove that the sum of these scores is not too large. Thus, when we use them for sampling, we will achieve column subsets and projection-cost preserving samples of small size.
Specifically we have:
\begin{lemma}\label{sumScores}
$\sum_{i=1}^n \bar \tau_i(\bv{A}) \le 2k.$
\end{lemma}
\begin{proof}
We rewrite \eqref{ridge_definition}
using $\bv{A}$'s singular value decomposition:
\begin{align*}
\bar \tau_i(\bv{A}) &= \bv{a}_i^T \left (\bv{U} \bv{\Sigma}^2 \bv{U}^T + \frac{\norm{\bv{A}_{\setminus k}}_F^2}{k} \bv{U}  \bv{U}^T \right)^+ \bv{a}_i\\
&= \bv{a}_i^T \left (\bv{U} \bv{\bar \Sigma}^2 \bv{U}^T \right )^+ \bv{a}_i
=  \bv{a}_i^T \left (\bv{U} \bv{\bar \Sigma}^{-2} \bv{U}^T \right ) \bv{a}_i,
\end{align*}
where $\bv{\bar \Sigma}^2_{i,i} = \sigma^2_i(\bv{A}) + \frac{\norm{\bv{A}-\bv{A}_k}_F^2}{k}$.
We then have:
\begin{align*}
\sum_{i=1}^n \bar \tau_i(\bv{A}) = \tr\left (\bv{A}^T \bv{U}\bv{\bar \Sigma}^{-2}\bv{U}^T \bv{A}\right ) = \tr \left (\bv{V} \bv{\Sigma} \bv{\bar \Sigma}^{-2} \bv{\Sigma} \bv{V}^T \right) = \tr(\bv{\Sigma}^2 \bv{\bar \Sigma}^{-2})
\end{align*}
$(\bv{\Sigma}^2 \bv{\bar \Sigma}^{-2})_{i,i} = \frac{\sigma_i^2(\bv{A})}{ \sigma^2_i(\bv{A}) + \frac{\norm{\bv{A}-\bv{A}_k}_F^2}{k}}$. For $i \le k$ we simply upper bound this by $1$. So:
\begin{align*}
\tr(\bv{\Sigma}^2 \bv{\bar \Sigma}^{-2}) = k + \sum_{i=k+1}^n\frac{\sigma_i^2}{\sigma_i^2 + \frac{\norm{\bv{A}-\bv{A}_k}_F^2}{k}} \leq k +\sum_{i=k+1}^n\frac{\sigma_i^2}{\frac{\norm{\bv{A}-\bv{A}_k}_F^2}{k}} = k +\frac{\sum_{i=k+1}^n\sigma_i^2}{\frac{\norm{\bv{A}-\bv{A}_k}_F^2}{k}} \leq k+ k.
\end{align*}
\end{proof}


\section{Core Sampling Results}
\label{sec:prelims}

Before considering how to efficiently compute ridge leverage scores, we prove that they can be used to construct sketches satisfying the guarantees of Definitions \ref{def:css} and \ref{def:pcp}. To do so, we introduce a natural intermediate guarantee (Theorem \ref{matrix_chernoff}), from which our results on column subset selection and projection-cost preservation follow. This approach is the first to treat these guarantees in a unified way and we hope it will be useful in future work on sketching methods for low-rank approximation.

Specifically, we will show that our selected columns \emph{spectrally approximate} 
$\bv{A}$ up to additive error depending on the ridge parameter $\lambda = \norm{\bv{A}-\bv{A}_k}_F/k$. 
This approximation is akin to the ubiquitous subspace embedding guarantee \cite{sarlos2006improved} which is used as a primitive for full rank problems like linear regression and generally requires sampling $\Theta(d)$ columns.

Intuitively, sampling by ridge leverage scores is equivalent to sampling by the standard leverage scores of $[\bv{A}, \sqrt{\lambda} \bv{I}_{n\times n}]$. 
A matrix Chernoff bound can be used to show that sampling by these scores will yield $\bv{C}$ satisfying the subspace embedding property: $(1-\epsilon)\bv{CC}^T \preceq \bv{A}\bv{A}^T + \lambda \bv{I} \preceq (1+\epsilon) \bv{CC}^T$. (Recall that $\bv{M} \preceq \bv{N}$ indicates that $\bv{s}^T\bv{M}\bv{s} \leq \bv{s}^T\bv{N}\bv{s}$ for every vector $\bv{s}$.) 

However, we do not actually sample columns of the identify, only columns of $\bv{A}$. Subtracting off the identity yields the mixed additive-multiplicative bound of Theorem \ref{matrix_chernoff}.

\begin{theorem}[Additive-Multiplicative Spectral Approximation]
\label{matrix_chernoff}
For $i \in \{1,\ldots,d\}$, let $\tilde \tau_i \ge \bar \tau_i(\bv{A})$ be an overestimate for the $i^{th}$ ridge leverage score. Let $p_i = \frac{\tilde \tau_i}{\sum_i \tilde \tau_i}$. Let $t = \frac{c\log(k/\delta)}{\epsilon^2} \sum_i \tilde \tau_i$ for some sufficiently large constant $c$. Construct $\bv{C}$ by sampling $t$ columns of $\bv{A}$, each set to $\frac{1}{\sqrt{tp_i}}\bv{a}_i$ with probability $p_i$. With probability $1-\delta$, $\bv{C}$ satisfies:
\begin{align}\label{freq_dirs_two_sided}
(1-\epsilon) \bv{C} \bv{C}^T - \frac{\epsilon}{k} \norm{\bv{A} - \bv{A}_k}_F^2 \bv{I}_{n\times n} \preceq \bv{A}\bv{A}^T \preceq (1+\epsilon) \bv{C} \bv{C}^T + \frac{\epsilon}{k} \norm{\bv{A} - \bv{A}_k}_F^2 \bv{I}_{n\times n}
\end{align}
\end{theorem}
By Lemma \ref{sumScores}, if each $\tilde \tau_i$ is within a constant factor of $\bar\tau_i(\bv{A})$ then $\bv{C}$ has $O\left(\frac{k\log(k/\delta)}{\epsilon^2}\right)$ columns. Note that Theorem \ref{matrix_chernoff} and our other sampling results hold for independent sampling without replacement. A proof is included in Appendix \ref{app:without_replacement}.
\begin{proof}
Following Lemma \ref{sumScores}, we have 
$\bar \tau_i(\bv{A}) =  \bv{a}_i^T \left (\bv{U} \bv{\bar \Sigma}^{-2} \bv{U}^T \right ) \bv{a}_i,
$
where $\bv{\bar \Sigma}^2_{i,i} = \sigma^2_i(\bv{A}) + \frac{\norm{\bv{A}_{\setminus k}}_F^2}{k}$.

Let $\bv{Y} = \bv{\bar \Sigma}^{-1}\bv{U}^T\left (\bv{CC}^T - \bv{AA}^T \right)\bv{U} \bv{\bar \Sigma}^{-1}$. We can write 
$$\bv{Y} = \sum_{j=1}^t \left [ \bv{\bar \Sigma}^{-1}\bv{U}^T\left (\bv{c}_j \bv{c}_j^T-\frac{1}{t}\bv{AA}^T\right )\bv{U} \bv{\bar \Sigma}^{-1}\right ]\eqdef \sum_{j=1}^t \left [ \bv{X}_j \right ].$$
For each $j \in 1,\ldots,t$, $\bv{X}_j$ is given by:
\begin{align*}
\bv{X}_j = 
\frac{1}{t} \cdot \bv{\bar \Sigma}^{-1}\bv{U} ^T\left (\frac{1}{p_i} \bv{a}_i\bv{a}_i^T - \bv{A}\bv{A}^T\right )\bv{U} \bv{\bar \Sigma}^{-1} \text{ with probability } p_i.
\end{align*}

$\E \bv{Y} = \bv{0}$ since $\E \left [ \frac{1}{p_i} \bv{a}_i\bv{a}_i^T - \bv{A}\bv{A}^T\right ] = \bv{0}$. Furthermore, $\bv{CC}^T = \bv{U}\bv{\bar \Sigma} \bv{Y} \bv{\bar \Sigma} \bv{U} + \bv{AA}^T$. Showing $\norm{\bv{Y}}_2 \le \epsilon$ gives $-\epsilon \bv{I} \preceq \bv{Y} \preceq  \epsilon \bv{I}$, and since $\bv{U}\bv{\bar \Sigma}^2 \bv{U}^T = \bv{AA}^T +\frac{ \norm{\bv{A}_{\setminus k}}_F^2}{k} \bv{I}$ would give:
\begin{align*}
(1-\epsilon)\bv{AA}^T - \frac{\epsilon \norm{\bv{A}_{\setminus k}}_F^2}{k} \bv{I} \preceq \bv{C}\bv{C}^T \preceq (1+\epsilon)\bv{AA}^T + \frac{\epsilon \norm{\bv{A}_{\setminus k}}_F^2}{k} \bv{I}.
\end{align*}
After rearranging and adjusting constants on $\epsilon$, this statement is equivalent to \eqref{freq_dirs_two_sided}.

To prove that $\norm{\bv{Y}}_2$ is small with high probability we use a stable rank (intrinsic dimension) matrix Bernstein inequality from \cite{tropp2015introduction} that was first proven in \cite{minsker} following work in \cite{stablechernoff}. This inequality requires upper bounds on the spectral norm of each $\bv{X}_j$ and on variance of $\bv{Y}$.

We use the fact that, for any $i$, $\frac{1}{\bar \tau_i(\bv{A})} \bv{a}_i \bv{a}_i^T \preceq \bv{A}\bv{A}^T + \frac{ \norm{\bv{A}_{\setminus k}}_F^2}{k} \bv{I}$. This is a well known property of leverage scores, shown for example in the proof of Lemma 11 in \cite{cohen2015uniform}. It lets us bound:
\begin{align*}
\frac{1}{\bar \tau_i(\bv{A})} \cdot \bv{\bar \Sigma}^{-1} \bv{U}^T\bv{a}_i \bv{a}_i^T \bv{U} \bv{\bar \Sigma}^{-1} \preceq \bv{\bar \Sigma}^{-1} \bv{U}^T\left(\bv{A}\bv{A}^T + \frac{ \norm{\bv{A}_{\setminus k}}_F^2}{k} \bv{I}\right) \bv{U} \bv{\bar \Sigma}^{-1} = \bv{I}.
\end{align*}
 So we have:
 \begin{align*}
\bv{X}_j + \frac{1}{t} \bv{\bar \Sigma}^{-1} \bv{U}^T\bv{AA}^T \bv{U} \bv{\bar \Sigma}^{-1} \preceq\frac{1}{tp_i} \cdot \bar \tau_i(\bv{A}) \cdot \bv{I} \preceq \frac{\epsilon^2}{c\log(k/\delta)\sum_i \tilde \tau_i} \cdot \frac{\sum_i \tilde \tau_i}{\tilde \tau_i} \cdot \bar \tau_i(\bv{A}) \cdot \bv{I} \preceq \frac{\epsilon^2}{c\log(k/\delta)} \bv{I}.
\end{align*}
Additionally,
\begin{align*}
\frac{1}{t} \bv{\bar \Sigma}^{-1} \bv{U}^T\bv{AA}^T \bv{U}\bv{\bar \Sigma}^{-1} = \frac{\epsilon^2}{c\log(k/\delta)\sum_i \tilde \tau_i}  \cdot \bv{\bar \Sigma}^{-2}\bv{\Sigma}^2 \preceq \frac{\epsilon^2}{c\log(k/\delta)} \bv{I},
\end{align*}
where the inequality follows from the fact that: $$\sum_i \tilde \tau_i \ge \sum_i \bar \tau_i(\bv{A}) = \tr \left (\bv{A}^T \bv{U}\bv{\bar \Sigma}^{-2} \bv{U}^T \bv{A} \right ) = \tr \left ( \bv{U} \bv{\bar \Sigma}^{-2}\bv{\Sigma}^2 \bv{U}^T\right) = \tr \left (\bv{\bar \Sigma}^{-2}\bv{\Sigma}^2 \right ) \ge \norm{\bv{\bar \Sigma}^{-2}\bv{\Sigma}^2}_2.$$
Overall this gives $\norm{\bv{X}_j}_2 \le \frac{\epsilon^2}{c\log(k/\delta)}$.
Next we bound the variance of $\bv{Y}$.
\begin{align}
\label{eq:var_bound}
\E (\bv{Y}^2) &= t \cdot \E (\bv{X}_j^2 ) = \frac{1}{t} \sum p_i \cdot  \left (\frac{1}{p_i^2} \bv{\bar \Sigma}^{-1} \bv{U}^T\bv{a}_i \bv{a}_i^T \bv{U} \bv{\bar \Sigma}^{-2} \bv{U}^T\bv{a}_i \bv{a}_i^T \bv{U} \bv{\bar \Sigma}^{-1} \right. \nonumber\\
&\left. - 2\frac{1}{p_i} \bv{\bar \Sigma}^{-1} \bv{U}^T\bv{a}_i \bv{a}_i^T \bv{U} \bv{\bar \Sigma}^{-2} \bv{U}^T\bv{AA}^T \bv{U} \bv{\bar \Sigma}^{-1} + \bv{\bar \Sigma}^{-1} \bv{U}^T\bv{AA}^T \bv{U} \bv{\bar \Sigma}^{-2} \bv{U}^T\bv{AA}^T \bv{U} \bv{\bar \Sigma}^{-1} \right) \nonumber\\
&\preceq \frac{1}{t}\sum \left [ \frac{\sum \tilde\tau_i}{\tilde \tau_i} \cdot  \bar \tau_i(\bv{A}) \cdot \bv{\bar \Sigma}^{-1} \bv{U}^T\bv{a}_i \bv{a}_i^T \bv{U} \bv{\bar \Sigma}^{-1}\right ] - \frac{1}{t}\bv{\bar \Sigma}^{-1} \bv{U}^T\bv{AA}^T \bv{U} \bv{\bar \Sigma}^{-2} \bv{U}^T\bv{AA}^T \bv{U} \bv{\bar \Sigma}^{-1}\nonumber\\
&\preceq \frac{\epsilon^2}{c\log(k/\delta)} \bv{\bar \Sigma}^{-1} \bv{U}^T\bv{A}\bv{A}^T \bv{U} \bv{\bar \Sigma}^{-1}\nonumber\\
&\preceq  \frac{\epsilon^2}{c\log(k/\delta)} \bv \Sigma^2\cdot \bv{\bar \Sigma}^{-2} 
\preceq \frac{\epsilon^2}{c\log(k/\delta)} \bv{D},
\end{align}
where we set $\bv{D}_{i,i} = 1$ for $i \in 1,\ldots,k$ and $\bv{D}_{i,i} = \frac{\sigma_i^2}{\sigma_i^2 + \frac{\norm{\bv{A}_{\setminus k}}_F^2}{k}}$ for all $i \in k+1,...,n$. By the stable rank matrix Bernstein inequality given in Theorem 7.3.1 of \cite{tropp2015introduction}, for $\epsilon < 1$,
\begin{align}
\label{eq:almost_chernoff}
\Pr \left [\norm{\bv Y} \ge \epsilon \right ] &\le \frac{4\tr(\bv D) }{\norm{\bv{D}}_2} \cdot e^{\frac{-\epsilon^2/2}{\left (\frac{\epsilon^2}{c\log(k/\delta)}(\norm{\bv{D}}_2+\epsilon/3)\right )}}.
\end{align}
Clearly $\norm{\bv D}_2 = 1$. Furthermore, following Lemma \ref{sumScores}, $\tr(\bv{D}) \le 2k$.
Plugging into \eqref{eq:almost_chernoff}, we see that
\begin{align*}
\Pr \left [\norm{\bv Y} \ge \epsilon \right ] &\le 8k e^{- \frac{c\log(k/\delta)}{2}})\le \delta/2,
\end{align*}
if we choose the constant $c$ large enough. 
So we have established \eqref{freq_dirs_two_sided}. 
\end{proof}

\subsection{Projection-Cost Preserving Sampling}

\label{sec:pcp}
We now use Theorem \ref{matrix_chernoff} to prove that sampling by ridge leverage scores is sufficient for constructing projection-cost preserving samples. The following theorem is a basic building block in our $O(\nnz(\bv{A}))$ time low-rank approximation algorithm.
\begin{theorem}[Projection-Cost Preservation]\label{pcp_intro} 
For $i \in \{1,\ldots,d\}$, let $\tilde \tau_i \ge \bar \tau_i(\bv{A})$ be an overestimate for the $i^{th}$ ridge leverage score. Let $p_i = \frac{\tilde \tau_i}{\sum_i \tilde \tau_i}$. Let $t = \frac{c\log(k/\delta)}{\epsilon^2} \sum_i \tilde \tau_i$ for any $\epsilon < 1$ and some sufficiently large constant $c$. Construct $\bv{C}$ by sampling $t$ columns of $\bv{A}$, each set to $\frac{1}{\sqrt{tp_i}}\bv{a}_i$ with probability $p_i$. With probability $1-\delta$, for any rank $k$ orthogonal projection $\bv{X}$,
\begin{align*}
(1-\epsilon)\|\bv{A} - \bv{X}\bv{A}\|^2_F \leq \|\bv{C} - \bv{X}\bv{C}\|^2_F\leq (1+\epsilon)\|\bv{A} - \bv{X}\bv{A}\|^2_F.
\end{align*}
\end{theorem}
Note that the theorem also holds for independent sampling without replacement, as shown in Appendix \ref{app:without_replacement}.
By Lemma \ref{sumScores}, when each approximation $\tilde\tau_i$ is within a constant factor of the true ridge leverage score $\bar \tau_i(\bv{A})$, we obtain a projection-cost preserving sample with $t = O(k\log(k/\delta)/\epsilon^2)$.

To simplify bookkeeping, we only worry about proving a version of Theorem \ref{pcp_intro} with $(1\pm a\epsilon)$ error for some constant $a$, and assume $\epsilon \le 1/2$. By simply adjusting our constant oversampling parameter, $c$, we can recover the result as stated.

The challenge in proving Theorem \ref{pcp_intro} comes from the mixed additive-multiplicative error of Theorem \ref{matrix_chernoff}. Pure multiplicative error, e.g. from a subspace embedding, or pure additive error, e.g. from a ``Frequent Directions'' sketch \cite{ghashami2015frequent}, are easily converted to projection-cost preservation results \cite{camsThesis}, but merging the analysis is intricate. To do so,
we split $\bv{A}\bv{A}^T$ and $\bv{C}\bv{C}^T$ into their projections onto the top ``head'' singular vectors of $\bv{A}$ and onto the remaining ``tail'' singular vectors. 
Restricted to the span of $\bv{A}$'s top singular vectors, Theorem \ref{matrix_chernoff} gives a purely multiplicative bound. Restricted to vectors spanned by $\bv{A}$'s lower singular vectors, the bound is purely additive.

\begin{proof}
For notational convenience, let $\bv{Y}$ denote $\bv{I}-\bv{X}$, so $\|\bv{A} - \bv{X}\bv{A}\|^2_F = \tr(\bv{Y}\bv{A}\bv{A}^T\bv{Y})$ and $\|\bv{C} - \bv{X}\bv{C}\|^2_F = \tr(\bv{Y}\bv{C}\bv{C}^T\bv{Y})$.

\subsubsection{Head/Tail Split}
Let $m$ be the index of the smallest singular value of $\bv{A}$ such that $\sigma_m^2 \geq \|\bv{A}-\bv{A}_{ k}\|_F^2/k$. 
Let $\bv{P}_m$ denote $\bv{U}_m\bv{U}_m^T$ and $\bv{P}_{\setminus m}$ denote $\bv{U}_{\setminus m}\bv{U}_{\setminus m}^T = \bv{I} - \bv{P}_m$. We split:
\begin{align} 
\tr(\bv{Y}\bv{A}\bv{A}^T\bv{Y}) &= \tr(\bv{Y}\bv{P}_m\bv{A}\bv{A}^T\bv{P}_m\bv{Y}) + \tr(\bv{Y}\bv{P}_{\setminus m}\bv{A}\bv{A}^T\bv{P}_{\setminus m}\bv{Y}) + 2\tr(\bv{Y}\bv{P}_m\bv{A}\bv{A}^T\bv{P}_{\setminus m}\bv{Y})
\nonumber\\
&= \tr(\bv{Y}\bv{A}_m\bv{A}_m^T\bv{Y}) + \tr(\bv{Y}\bv{A}_{\setminus m}\bv{A}_{\setminus m}^T\bv{Y}).  
\label{eq:a_split}
\end{align}
The ``cross terms'' involving $\bv{P}_m\bv{A}$ and $\bv{P}_{\setminus m}\bv{A}$ equal $0$ since the two matrices have mutually orthogonal rows (spanned by $\bv{V}_m^T$ and $\bv{V}_{\setminus m}^T$, respectively). Additionally, we split:
\begin{align} 
\tr(\bv{Y}\bv{C}\bv{C}^T\bv{Y}) &= \tr(\bv{Y}\bv{P}_m\bv{C}\bv{C}^T\bv{P}_m\bv{Y}) + \tr(\bv{Y}\bv{P}_{\setminus m}\bv{C}\bv{C}^T\bv{P}_{\setminus m}\bv{Y})  + 2\tr(\bv{Y}\bv{P}_{m}\bv{C}\bv{C}^T\bv{P}_{\setminus m}\bv{Y}) \label{eq:c_split}
\end{align}
In \eqref{eq:c_split} cross terms do not cancel because, in general, $\bv{P}_{ m}\bv{C}$ and $\bv{P}_{\setminus m}\bv{C}$ \emph{will not} have orthogonal rows, even though they have orthogonal columns. Regardless, while these terms make our analysis more difficult, we proceed with comparing corresponding parts of \eqref{eq:a_split} and \eqref{eq:c_split}.

\subsubsection{Head Terms}
We first bound the terms involving $\bv{P}_{m}$, beginning by showing that:
\begin{align}
\label{eq:pure_mult}
\frac{1- \epsilon}{1+\epsilon} \bv{P}_m\bv{C}\bv{C}^T\bv{P}_m \preceq \bv{A}_m\bv{A}_m^T \preceq \frac{1+ \epsilon}{1-\epsilon} \bv{P}_m\bv{C}\bv{C}^T\bv{P}_m.
\end{align}
For any vector $\bv{x}$, let $\bv{y} = \bv{P}_m \bv{x}$. 
Note that $\bv{x}^T\bv{A}_m\bv{A}_m^T\bv{x} = \bv{y}^T\bv{A}\bv{A}^T\bv{y}$ since $\bv{A}_m\bv{A}_m^T = \bv{P}_m\bv{A}\bv{A}^T\bv{P}_m$ and since $\bv{P}_m\bv{P}_m = \bv{P}_m$. So, using \eqref{freq_dirs_two_sided} 
we can bound:
\begin{align}
\label{eq:raw_top_bound}
(1-\epsilon)\bv{y}^T\bv{C}\bv{C}^T\bv{y} - \epsilon\frac{\|\bv{A}_{\setminus k}\|_F^2}{k}\bv{y}^T\bv{y}  \leq \bv{x}^T\bv{A}_m\bv{A}_m^T\bv{x} \leq (1+\epsilon)\bv{y}^T\bv{C}\bv{C}^T\bv{y} + \epsilon\frac{\|\bv{A}_{\setminus k}\|_F^2}{k}\bv{y}^T\bv{y}.
\end{align}
By our definition of $m$, $\bv{y}$ is orthogonal to all singular directions of $\bv{A}$ except those with squared singular value greater than or equal to $\|\bv{A}_{\setminus k}\|_F^2/k$. It follows that 
\begin{align*}
\bv{x}^T\bv{A}_m\bv{A}_m^T\bv{x} = \bv{y}^T\bv{A}\bv{A}^T\bv{y} \geq \frac{\|\bv{A}_{\setminus k}\|_F^2}{k}\bv{y}^T\bv{y},
\end{align*}
and accordingly, from the left side of \eqref{eq:raw_top_bound}, that $(1-\epsilon)\bv{y}^T\bv{C}\bv{C}^T\bv{y} \leq (1+\epsilon)\bv{x}^T\bv{A}_m\bv{A}_m^T\bv{x} $. 
Additionally, from the right side of \eqref{eq:raw_top_bound}, we have that $(1+\epsilon)\bv{y}^T\bv{C}\bv{C}^T\bv{y} \geq (1-\epsilon)\bv{x}^T\bv{A}_m\bv{A}_m^T\bv{x}$. Since $\bv{y}^T\bv{C}\bv{C}^T\bv{y} = \bv{x}^T\bv{P}_m\bv{C}\bv{C}^T\bv{P}_m\bv{x}$, 
these inequalities combine to prove \eqref{eq:pure_mult}. From \eqref{eq:pure_mult} we can bound the diagonal entries of $\bv{Y}\bv{A}_m\bv{A}_m^T\bv{Y}$ in terms of the corresponding diagonal entries of $\bv{Y}\bv{P}_m\bv{C}\bv{C}^T\bv{P}_m\bv{Y}$, which are all positive, and conclude that:
\begin{align*}
\frac{1- \epsilon}{1+\epsilon}\tr(\bv{Y}\bv{P}_m\bv{C}\bv{C}^T\bv{P}_m\bv{Y}) \leq  \tr(\bv{Y}\bv{A}_m\bv{A}_m^T\bv{Y}) \leq \frac{1+ \epsilon}{1-\epsilon}\tr(\bv{Y}\bv{P}_m\bv{C}\bv{C}^T\bv{P}_m\bv{Y}).
\end{align*}
Assuming $\epsilon < 1/2$, this is equivalent to:
\begin{align}
(1-4\epsilon)\tr(\bv{Y}\bv{A}_m\bv{A}_m^T\bv{Y}) \leq \tr(\bv{Y}\bv{P}_m\bv{C}\bv{C}^T\bv{P}_m\bv{Y}) \leq (1+4\epsilon)\tr(\bv{Y}\bv{A}_m\bv{A}_m^T). \label{eq:top_bound}
\end{align}

\subsubsection{Tail Terms}
For the lower singular directions of $\bv{A}$, Theorem \ref{matrix_chernoff} does not give a multiplicative spectral approximation, so we do things a bit differently. Specifically, we start by noting that:
\begin{align*}
\tr(\bv{Y}\bv{A}_{\setminus m}\bv{A}_{\setminus m}^T\bv{Y}) &= \tr(\bv{A}_{\setminus m}\bv{A}_{\setminus m}^T) - \tr(\bv{X}\bv{A}_{\setminus m}\bv{A}_{\setminus m}^T\bv{X}) \text{ and } \\
\tr(\bv{Y}\bv{P}_{\setminus m}\bv{C}\bv{C}^T\bv{P}_{\setminus m}\bv{Y}) &= \tr(\bv{P}_{\setminus m}\bv{C}\bv{C}^T\bv{P}_{\setminus m}) -  \tr(\bv{X}\bv{P}_{\setminus m}\bv{C}\bv{C}^T\bv{P}_{\setminus m}\bv{X}). 
\end{align*}
We handle $\tr(\bv{A}_{\setminus m}\bv{A}_{\setminus m}^T) = \|\bv{A}_{\setminus m}\|_F^2$ and $\tr(\bv{P}_{\setminus m}\bv{C}\bv{C}^T\bv{P}_{\setminus m}) = \|\bv{P}_{\setminus m}\bv{C}\|_F^2$ first. Since $\bv{C}$ is constructed via an unbiased sampling of $\bv{A}$'s columns, $\E\left[\|\bv{P}_{\setminus m}\bv{C}\|_F^2\right] = \|\bv{A}_{\setminus m}\|_F^2$ and a scalar Chernoff bound is sufficient for showing that this value concentrates around its expectation. Our proof is included as Lemma \ref{trace:bound} in Appendix \ref{app:trace_bound} and implies the following bound:
\begin{align}
\label{tail_mat_trace_bound}
-\epsilon\|\bv{A}_{\setminus k}\|_F^2 \leq \tr(\bv{A}_{\setminus m}\bv{A}_{\setminus m}^T) - \tr(\bv{P}_{\setminus m}\bv{C}\bv{C}^T\bv{P}_{\setminus m}) \leq \epsilon\|\bv{A}_{\setminus k}\|_F^2.
\end{align}
Next, we compare $\tr(\bv{X}\bv{A}_{\setminus m}\bv{A}_{\setminus m}^T\bv{X})$ to  $\tr(\bv{X}\bv{P}_{\setminus m}\bv{C}\bv{C}^T\bv{P}_{\setminus m}\bv{X})$. We first claim that:
\begin{align}
\label{eq:pure_add}
\bv{P}_{\setminus m}\bv{C}\bv{C}^T\bv{P}_{\setminus m} - \frac{4\epsilon}{k} \norm{\bv{A}_{\setminus k}}_F^2\bv{I} \preceq \bv{A}_{\setminus m}\bv{A}_{\setminus m}^T \preceq \bv{P}_{\setminus m}\bv{C}\bv{C}^T\bv{P}_{\setminus m} +  \frac{4\epsilon}{k} \norm{\bv{A}_{\setminus k}}_F^2\bv{I}.
\end{align}
The argument is similar to the one for \eqref{eq:pure_mult}.
For a vector $\bv{x}$, let $\bv{y} = \bv{P}_{\setminus m} \bv{x}$. 
$\bv{x}^T\bv{A}_{\setminus m}\bv{A}_{\setminus m}^T\bv{x} = \bv{y}^T\bv{A}\bv{A}^T\bv{y}$ since $\bv{A}_{\setminus m}\bv{A}_{\setminus m}^T = \bv{P}_{\setminus m}\bv{A}\bv{A}^T\bv{P}_{\setminus m}$ and since $\bv{P}_{\setminus m}\bv{P}_{\setminus m} = \bv{P}_{\setminus m}$. Applying \eqref{freq_dirs_two_sided} gives:
\begin{align*}
(1-\epsilon)\bv{y}^T\bv{C}\bv{C}^T\bv{y} - \epsilon\frac{\|\bv{A}_{\setminus k}\|_F^2}{k}\bv{y}^T\bv{y}  \leq \bv{x}^T\bv{A}_{\setminus m}\bv{A}_{\setminus m}^T\bv{x} \leq (1+\epsilon)\bv{y}^T\bv{C}\bv{C}^T\bv{y} + \epsilon\frac{\|\bv{A}_{\setminus k}\|_F^2}{k}\bv{y}^T\bv{y}.
\end{align*}
Noting that $\bv{y}^T\bv{y} \leq \bv{x}^T\bv{x}$ and assuming $\epsilon \leq 1/2$ gives the following two inequalities:
\begin{align}
\bv{y}^T\bv{C}\bv{C}^T\bv{y} - 2\epsilon\frac{\|\bv{A}_{\setminus k}\|_F^2}{k}\bv{x}^T\bv{x} &\leq (1+2\epsilon)\bv{x}^T\bv{A}_{\setminus m}\bv{A}_{\setminus m}^T\bv{x}, \label{eq:raw_tail_bound_1}\\
(1-2\epsilon)\bv{x}^T\bv{A}_{\setminus m}\bv{A}_{\setminus m}^T\bv{x} &\leq \bv{y}^T\bv{C}\bv{C}^T\bv{y} + 2\epsilon\frac{\|\bv{A}_{\setminus k}\|_F^2}{k}\bv{x}^T\bv{x} \label{eq:raw_tail_bound_2}.
\end{align}
By our choice of $m$, $\bv{x}^T\bv{A}_{\setminus m}\bv{A}_{\setminus m}^T\bv{x} \leq \frac{\|\bv{A}_{\setminus k}\|_F^2}{k}\bv{x}^T\bv{x}$. So, substituting $\bv{y}$ with $\bv{P}_{\setminus m} \bv{x}$ and rearranging \eqref{eq:raw_tail_bound_1} and \eqref{eq:raw_tail_bound_2} gives \eqref{eq:pure_add}. 

Now, since $\bv{X}$ is a rank $k$ projection matrix, it can be written as $\bv{X} = \bv{Z}\bv{Z}^T$ where $\bv{Z} \in \R^{n\times k}$ is a matrix with $k$ orthonormal columns, $\bv{z}_1,\ldots,\bv{z}_k$.  By cyclic property of the trace, 
\begin{align*}
\tr(\bv{X}\bv{A}_{\setminus m}\bv{A}_{\setminus m}^T\bv{X}) = \tr(\bv{Z}^T\bv{A}_{\setminus m}\bv{A}_{\setminus m}^T\bv{Z}) = \sum_{i=1}^k \bv{z}_i^T\bv{A}_{\setminus m}\bv{A}_{\setminus m}^T\bv{z}_i.
\end{align*}
Similarly, $\tr(\bv{X}\bv{P}_{\setminus m}\bv{C}\bv{C}^T\bv{P}_{\setminus m}\bv{X}) = \sum_{i=1}^k \bv{z}_i^T\bv{P}_{\setminus m}\bv{C}\bv{C}^T\bv{P}_{\setminus m}\bv{z}_i$ and we conclude from \eqref{eq:pure_add} that:
\begin{align*}
\tr(\bv{X}\bv{P}_{\setminus m}\bv{C}\bv{C}^T\bv{P}_{\setminus m}\bv{X}) - 4\epsilon\|\bv{A}_{\setminus k}\|_F^2 \leq \tr(\bv{X}\bv{A}_{\setminus m}\bv{A}_{\setminus m}^T\bv{X})  \leq \tr(\bv{X}\bv{P}_{\setminus m}\bv{C}\bv{C}^T\bv{P}_{\setminus m}\bv{X}) + 4\epsilon\|\bv{A}_{\setminus k}\|_F^2,
\end{align*}
which combines with \eqref{tail_mat_trace_bound} to give the final bound:
\begin{align}
\tr(\bv{Y}\bv{A}_{\setminus m}\bv{A}_{\setminus m}^T\bv{Y}) - 5\epsilon\|\bv{A}_{\setminus k}\|_F^2 \leq \tr(\bv{Y}\bv{P}_{\setminus m}\bv{C}\bv{C}^T\bv{P}_{\setminus m}\bv{Y}) \leq \tr(\bv{Y}\bv{A}_{\setminus m}\bv{A}_{\setminus m}^T\bv{Y}) + 5\epsilon\|\bv{A}_{\setminus k}\|_F^2.\label{eq:bottom_bound}
\end{align}

\subsubsection{Cross Terms}
Finally, we handle the cross term $2\tr(\bv{Y}\bv{P}_{m}\bv{C}\bv{C}^T\bv{P}_{\setminus m}\bv{Y})$. We do not have anything to compare this term to, so we just need to show that it is small. To do so, we rewrite:
\begin{align}
\tr(\bv{Y}\bv{P}_m\bv{C}\bv{C}^T\bv{P}_{\setminus m}\bv{Y}) = \tr(\bv{Y}\bv{A}\bv{A}^T(\bv{A}\bv{A}^T)^+\bv{P}_{m}\bv{C}\bv{C}^T\bv{P}_{\setminus m}),
\end{align}
which is an equality since the columns of $\bv{P}_{m}\bv{C}\bv{C}^T\bv{P}_{\setminus m}$ fall in the span of $\bv{A}$'s columns. We eliminate the trailing $\bv{Y}$ using the cyclic property of the trace. $\langle \bv{M}, \bv{N} \rangle = \tr( \bv{M} (\bv{A}\bv{A}^T)^+ \bv{N}^T)$ is a semi-inner product since $\bv{A}\bv{A}^T$ is positive semidefinite. Thus, by the Cauchy-Schwarz inequality,
\begin{align}
|\tr(\bv{Y}\bv{A}\bv{A}^T(\bv{A}\bv{A}^T)^+&\bv{P}_{m}\bv{C}\bv{C}^T\bv{P}_{\setminus m})| \leq \nonumber \\
&\sqrt{\tr(\bv{Y}\bv{A}\bv{A}^T(\bv{A}\bv{A}^T)^+\bv{A}\bv{A}^T\bv{Y})\cdot\tr(\bv{P}_{\setminus m}\bv{C}\bv{C}^T\bv{P}_{m}(\bv{A}\bv{A}^T)^+\bv{P}_{m}\bv{C}\bv{C}^T\bv{P}_{\setminus m})} \nonumber\\
&= \sqrt{\tr(\bv{Y}\bv{A}\bv{A}^T\bv{Y})\cdot\tr(\bv{P}_{\setminus m}\bv{C}\bv{C}^T\bv{U}_{m}\bv{\Sigma}_{m}^{-2}\bv{U}_{m}^T\bv{C}\bv{C}^T\bv{P}_{\setminus m})} \nonumber\\
&= \sqrt{\tr(\bv{Y}\bv{A}\bv{A}^T\bv{Y})}\cdot \sqrt{\|\bv{P}_{\setminus m}\bv{C}\bv{C}^T\bv{U}_{m}\bv{\Sigma}_{m}^{-1}\|_F^2} \label{eq:cauchy_s}.
\end{align}
To bound the second term, we separate:
\begin{align}
\label{eq:sum_sep}
\|\bv{P}_{\setminus m}\bv{C}\bv{C}^T\bv{U}_{m}\bv{\Sigma}_{m}^{-1}\|_F^2 = \sum_{i=1}^m \|\bv{P}_{\setminus m}\bv{C}\bv{C}^T\bv{u}_{i}\|_2^2\sigma_{i}^{-2}.
\end{align}
We next show that the summand is small for every $i$. Take $\bv{p}_i$ to be a unit vector in the direction of $\bv{C}\bv{C}^T\bv{u}_{i}$'s projection onto $\bv{P}_{\setminus m}$.  I.e. $\bv{p}_i = \bv{P}_{\setminus m}\bv{C}\bv{C}^T\bv{u}_{i} / \| \bv{P}_{\setminus m}\bv{C}\bv{C}^T\bv{u}_{i}\|_2$. Then:
\begin{align}
\|\bv{P}_{\setminus m}\bv{C}\bv{C}^T\bv{u}_{i}\|_2^2 = (\bv{p}_i^T\bv{C}\bv{C}\bv{u}_i)^2.
\end{align}
Now, suppose we construct the vector $\bv{m} = \left(\sigma_i^{-1}\bv{u}_i + \frac{\sqrt{k}}{\|\bv{A}_{\setminus k}\|_F}\bv{p}_i\right)$. From \eqref{freq_dirs_two_sided} we know that:
\begin{align*}
(1-\epsilon)\bv{m}^T\bv{C}\bv{C}^T\bv{m} - \frac{\epsilon\|\bv{A}_{\setminus k}\|_F^2}{k} \bv{m}^T\bv{m} &\leq \bv{m}^T\bv{A}\bv{A}^T\bv{m},
\end{align*}
which expands to give:
\begin{align}
(1-\epsilon)\sigma_i^{-2}\bv{u}_i^T\bv{C}\bv{C}^T\bv{u}_i + (1-\epsilon)\frac{k}{\|\bv{A}_{\setminus k}\|^2_F} \bv{p}_i^T\bv{C}\bv{C}^T\bv{p}_i + (1-\epsilon)\frac{2\sqrt{k}}{\sigma_i \|\bv{A}_{\setminus k}\|_F}\bv{p}_i^T\bv{C}\bv{C}^T\bv{u}_i \nonumber \leq\\
 \sigma_i^{-2}\bv{u}_i^T\bv{A}\bv{A}^T\bv{u}_i + \frac{k}{\|\bv{A}_{\setminus k}\|^2_F} + \frac{\epsilon\|\bv{A}_{\setminus k}\|_F^2}{k} \bv{m}^T\bv{m}
= 1+ \frac{k}{\|\bv{A}_{\setminus k}\|^2_F}\bv{p}_i^T\bv{A}\bv{A}^T\bv{p}_i + \frac{\epsilon\|\bv{A}_{\setminus k}\|_F^2}{k} \bv{m}^T\bv{m}. 
\end{align}
There are no cross terms on the right side because $\bv{p}_i$ lies in the span of $\bv{U}_{\setminus m}$ and is thus orthogonal to $\bv{u}_i$ over $\bv{A}\bv{A}^T$.
Now, from \eqref{eq:pure_mult} we know that  $\bv{u}_i^T\bv{C}\bv{C}^T\bv{u}_i \geq (1-2\epsilon)\bv{u}_i^T\bv{A}\bv{A}^T\bv{u}_i \geq (1-2\epsilon)\sigma_i^2$. From  \eqref{eq:pure_add} we also know that $\bv{p}_i^T\bv{C}\bv{C}^T\bv{p}_i \geq \bv{p}_i^T\bv{A}\bv{A}^T\bv{p}_i - 4\epsilon\frac{\|\bv{A}_{\setminus k}\|^2_F}{k}$. Plugging into \eqref{eq:cross_term_bound} gives:
\begin{align}
&(1-3\epsilon)\sigma_i^{-2}\bv{u}_i^T\bv{A}\bv{A}^T\bv{u}_i + (1-\epsilon)\frac{k}{\|\bv{A}_{\setminus k}\|^2_F} \bv{p}_i^T\bv{A}\bv{A}^T\bv{p}_i -4\epsilon + (1-\epsilon)\frac{2\sqrt{k}}{\sigma_i \|\bv{A}_{\setminus k}\|_F}\bv{p}_i^T\bv{C}\bv{C}^T\bv{u}_i \nonumber \\
&\leq 1+ \frac{k}{\|\bv{A}_{\setminus k}\|^2_F}\bv{p}_i^T\bv{A}\bv{A}^T\bv{p}_i + \frac{\epsilon\|\bv{A}_{\setminus k}\|_F^2}{k} \bv{m}^T\bv{m}. \label{eq:cross_term_bound_2}
\end{align}
Noting that $\bv{p}_i^T\bv{A}\bv{A}^T\bv{p}_i \leq \frac{\|\bv{A}_{\setminus k}\|_F^2}{k}$ since $\bv{p}_i$ lies in the column span of $\bv{U}_{\setminus m}$, rearranging \eqref{eq:cross_term_bound_2} gives:
\begin{align*}
(1-\epsilon)\frac{2\sqrt{k}}{\sigma_i \|\bv{A}_{\setminus k}\|_F}\bv{p}_i^T\bv{C}\bv{C}^T\bv{u}_i \leq 8\epsilon + \frac{\epsilon\|\bv{A}_{\setminus k}\|_F^2}{k} \bv{m}^T\bv{m} \leq 12 \epsilon.
\end{align*}
The second inequality follows from the fact that $\sigma_i^{-1} \leq \frac{\sqrt{k}}{\|\bv{A}_{\setminus k}\|_F}$ so $\|\bv{m}\|_2^2 \leq \left(\frac{2\sqrt{k}}{\|\bv{A}_{\setminus k}\|_F}\right)^2.$
Assuming again that $\epsilon \leq 1/2$ gives our final bound:
\begin{align}
\frac{\sqrt{k}}{\sigma_i \|\bv{A}_{\setminus k}\|_F}\bv{p}_i\bv{C}\bv{C}^T\bv{u}_i^T \leq 12\epsilon \nonumber\\
(\bv{p}_i\bv{C}\bv{C}^T\bv{u}_i^T)^2 \leq 144\epsilon^2\frac{\sigma_i^2 \|\bv{A}_{\setminus k}\|^2_F}{k}.
\end{align}
Plugging into \eqref{eq:sum_sep} gives:
\begin{align}
\|\bv{P}_{\setminus m}\bv{C}\bv{C}^T\bv{U}_{m}\bv{\Sigma}_{m}^{-1}\|_F^2 \leq \sum_{i=1}^m 144\epsilon^2\frac{\sigma_i^2 \|\bv{A}_{\setminus k}\|^2_F}{k}\sigma_{i}^{-2} \leq 288\epsilon^2 \|\bv{A}_{\setminus k}\|^2_F.
\end{align}
Note that we get an extra factor of 2 because $m \leq 2k$. Returning to \eqref{eq:cauchy_s}, we conclude that:
\begin{align}
|\tr(\bv{Y}\bv{A}\bv{A}^T(\bv{A}\bv{A}^T)^+\bv{P}_{m}\bv{C}\bv{C}^T\bv{P}_{\setminus m})| &\leq \sqrt{\tr(\bv{Y}\bv{A}\bv{A}^T\bv{Y})}\cdot \sqrt{288\epsilon^2 \|\bv{A}_{\setminus k}\|^2_F} \leq 17\epsilon\tr(\bv{Y}\bv{A}\bv{A}^T\bv{Y}).
\label{eq:cross_term_bound}
\end{align}
The last inequality follows from the fact that $\|\bv{A}_{\setminus k}\|^2_F \leq \tr(\bv{Y}\bv{A}\bv{A}^T\bv{Y})$ since $\bv{A}_{\setminus k}$ is the best rank $k$ approximation to $\bv{A}$. $\tr(\bv{Y}\bv{A}\bv{A}^T\bv{Y}) = \|\bv{A} - \bv{X}\bv{A}\|_F^2$ is the error of a suboptimal rank $k$ approximation.

\subsubsection{Final Bound}
Ultimately, from\eqref{eq:c_split}, \eqref{eq:top_bound}, \eqref{eq:bottom_bound}, and \eqref{eq:cross_term_bound}, we conclude:
\begin{align*}
(1-4\epsilon)\tr(\bv{Y}\bv{A}_m\bv{A}_m^T\bv{Y}) + \tr(\bv{Y}\bv{A}_{\setminus m}\bv{A}_{\setminus m}^T\bv{Y}) - 5\epsilon\|\bv{A}_{\setminus k}\|^2_F -  34\epsilon\tr(\bv{Y}\bv{A}\bv{A}^T\bv{Y})  \leq \tr(\bv{Y}\bv{C}\bv{C}^T\bv{Y}) \nonumber \\
\leq  (1+4\epsilon)\tr(\bv{Y}\bv{A}_m\bv{A}_m^T) + \tr(\bv{Y}\bv{A}_{\setminus m}\bv{A}_{\setminus m}^T\bv{Y}) +5\epsilon\|\bv{A}_{\setminus k}\|^2_F + 34\tr(\bv{Y}\bv{A}\bv{A}^T\bv{Y}).
\end{align*}
Applying the fact that $\|\bv{A}_{\setminus k}\|^2_F \leq \tr(\bv{Y}\bv{A}\bv{A}^T\bv{Y})$ proves Theorem \ref{pcp_intro} for a constant factor of $\epsilon$. \end{proof}

\subsection{Column Subset Selection}
\label{sec:css} 
Although not required for our main low-rank approximation algorithm, we also prove that ridge leverage score sampling can be used to obtain $(1+\epsilon)$ error column subsets (Definition \ref{def:css}). The column subset selection problem is of independent interest and the following result allows ridge leverage scores to be used in our single-pass streaming algorithm for this problem (Section \ref{sec:streaming}).

\begin{theorem}\label{css_intro} 
For $i \in \{1,\ldots,d\}$, let $\tilde \tau_i \ge \bar \tau_i(\bv{A})$ be an overestimate for the $i^{th}$ ridge leverage score. Let $p_i = \frac{\tilde \tau_i}{\sum_i \tilde \tau_i}$. Let $t = c \left(\log k + \frac{\log(1/\delta)}{\epsilon}\right) \sum_i \tilde \tau_i$ for $\epsilon < 1$ and some sufficiently large constant $c$. Construct $\bv{C}$ by sampling $t$ columns of $\bv{A}$, each set to $\bv{a}_i$ with probability $p_i$. With probability $1-\delta$:
\begin{align*}
\norm{\bv{A}-\left (\bv{C}\bv{C}^+\bv{A}\right)_k}_F^2 \le (1+\epsilon) \norm{\bv{A}-\bv{A}_k}_F^2.
\end{align*}
Furthermore, $\bv{C}$ contains a subset of $O\left(\sum_i \tilde \tau_i/\epsilon \right)$ columns that satisfies Definition \ref{def:css} and can be identified in polynomial time.
\end{theorem}
Note that $\left(\bv{C}\bv{C}^+\bv{A}\right)_k$ is a rank $k$ matrix in the column span of $\bv{C}$, so Theorem \ref{css_intro} implies that $\bv{C}$ is a $(1+\epsilon)$ error column subset according to Definition \ref{def:css}.
By Lemma \ref{sumScores}, if each $\tilde \tau_i$ is within a constant factor of $\bar \tau_i(\bv{A})$, the approximate ridge leverage scores sum to $O(k)$ so Theorem \ref{css_intro} gives a column subset of size $O\left (k\log k + k/\epsilon \right)$, which contains a near optimally sized column subset with $O\left (k/\epsilon \right )$ columns. Again, the theorem also holds for sampling without replacement (see Appx. \ref{app:without_replacement}).


Our proof relies on establishing a connection between ridge leverage sampling and well known \emph{adaptive sampling} techniques for column subset selection 
\cite{deshpande2006matrix,deshpandeVempala}. 
We start with the following lemma on adaptive sampling for column subset selection:

\begin{lemma}[Theorem 2.1 of \cite{deshpande2006matrix}]\label{adaptive_sampling_primative}
Let $\bv{C}$ be any subset of $\bv{A}$'s columns and let $\bv{Z}$ be an orthonormal matrix whose columns span those of $\bv{C}$. If we sample an additional set $\bv{S}$ of $O \left (\frac{k\log(1/\delta)}{\epsilon} \cdot \frac{\norm{\bv{A}-\bv{Z}\bv{Z}^T\bv{A}}_F^2}{\norm{\bv{A}_{\setminus k}}_F^2} \right )$ columns from $\bv{A}$ with probability proportional to $\norm{(\bv{A}-\bv{Z}\bv{Z}^T\bv{A})_i}_2^2$, then $[\bv{S}\cup\bv{C}]$ is a $(1+\epsilon)$ error column subset for $\bv{A}$ with probability $(1-\delta)$. \footnote{Theorem 2.1 was originally stated as an expected error result, but it can be seen to hold with constant probability via Markov's inequality and accordingly with $(1-\delta)$ probability when oversampling by a factor of $\log(1/\delta)$}
\end{lemma}

When $\bv{C}$ is a constant error column subset, then $\norm{\bv{A}-\bv{Z}\bv{Z}^T\bv{A}}_F^2 \leq \norm{\bv{A}-(\bv{Z}\bv{Z}^T\bv{A})_k}_F^2 = O(\norm{\bv{A}_{\setminus k}}_F^2)$ and accordingly we only need $O (k\log(1/\delta)/\epsilon)$ additional adaptive samples. So one potential algorithm for column subset selection is as follows: apply Theorems \ref{matrix_chernoff} and \ref{pcp_intro}, sampling $O(k\log(k/\delta))$ columns by ridge leverage score to obtain a constant error projection-cost preserving sample, will also be a constant error column subset. Then sample $O (k\log(1/\delta)/\epsilon)$ additional columns adaptively against $\bv{C}$.

However, it turns out that ridge leverage scores well approximate adaptive sampling probabilities computed with respect to \emph{any} constant error additive-multiplicative spectral approximation satisfying Theorem \ref{matrix_chernoff}! That is, surprisingly, they achieve the performance of adaptive sampling without being adaptive at all. Simply sampling $O(k \log(1/\delta)/\epsilon)$ more columns by ridge leverage score and invoking Lemma \ref{adaptive_sampling_primative} suffices to achieve $(1+\epsilon)$ error.

\begin{proof}[Proof of Theorem \ref{css_intro}]
 We formally prove that $\bv{C}$ is itself a good column subset before showing our stronger guarantee, that it also contains a column subset of optimal size, up to constants.
\subsubsection{Primary Column Subset Selection Guarantee}
We split our sample $\bv{C}$, into $\bv{C}_1$, which contains the first $c \log (k/\delta) \sum_i \tilde \tau_i$ columns and $\bv{C}_2$, which contains the next $c \log (1/\delta)/\epsilon \sum_i \tilde \tau_i$ columns. Note that in our final sample complexity the $\log (1/\delta)$ factor in the size of $\bv{C}_1$ is not shown as it is absorbed into the larger size of $\bv{C}_2$ when $\log(1/\delta) > \log(k)$ and into the $\log(k)$ otherwise.
By Theorem \ref{pcp_intro}, we know that, appropriately reweighted, $\bv{C}_1$ is a constant error projection-cost preserving sample of $\bv{A}$. This means that $\bv{C}_1$ is also a constant error 
column subset. Let $\bv{Z}$ be an orthonormal matrix whose columns span the columns of $\bv{C}_1$.

To invoke Lemma \ref{adaptive_sampling_primative} to boost $\bv{C}_1$ to a $(1+\epsilon)$ column subset, we need to sample columns with probabilities proportional to $\norm{(\bv{A}-\bv{Z}\bv{Z}^T\bv{A})_i}_2^2$. This is equivalent to sampling proportional to:
\begin{align*}
(\bv{a}_i^T - \bv{a}_i^T \bv{ZZ}^T)(\bv{a}_i - \bv{ZZ}^T\bv{a}_i ) = \bv{a}_i^T\bv{a}_i - 2\bv{a}_i^T \bv{ZZ}^T \bv{a}_i + \bv{a}_i^T \bv{ZZ}^T\bv{ZZ}^T \bv{a}_i =  \bv{a}_i^T \left (\bv{I} - \bv{ZZ}^T \right ) \bv{a}_i.
\end{align*}
We can assume $\|\bv{A}_{\setminus k}\|_F^2 > 0$ or else $\bv{C}_1$ must fully span $\bv{A}$'s columns and we're done. Scaling $ \bar \tau_i(\bv{A})$:
\begin{align*}
\frac{\norm{\bv{A}_{\setminus k}}_F^2}{k} \bar \tau_i(\bv{A}) = \bv{a}_i^T\left (\frac{k}{\norm{\bv{A}_{\setminus k}}_F^2} \bv{AA}^T + \bv{I} \right )^+ \bv{a}_i.
\end{align*}
Since $\bv{C}_1$ satisfies Theorem \ref{matrix_chernoff} with constant error, for large enough constant $c_1$,
\begin{align*}
\frac{k}{\norm{\bv{A}_{\setminus k}}_F^2} \bv{AA}^T + \bv{I} \preceq c_1 \left (\frac{k}{\norm{\bv{A}_{\setminus k}}_F^2} \bv{C}_1\bv{C}_1^T + \bv{I} \right )
\preceq c_1 \left (\bv{I} + \frac{k\norm{\bv{C}_1\bv{C}_1^T}_2}{\norm{\bv{A}_{\setminus k}}_F^2} \bv{ZZ}^T \right ).
\end{align*}
Furthermore, $\bv{I}-\bv{ZZ}^T \preceq \left (\bv{I} + c\bv{ZZ}^T \right )^+$ for \emph{any} positive $c$ so,
\begin{align*}
c_1\left (\frac{k}{\norm{\bv{A}_{\setminus k}}_F^2} \bv{AA}^T + \bv{I} \right )^+ \succeq \left (\bv{I} + \frac{k\norm{\bv{C}_1\bv{C}_1^T}_2}{\norm{\bv{A}_{\setminus k}}_F^2} \bv{ZZ}^T \right )^+ \succeq \bv{I}-\bv{ZZ}^T.
\end{align*}
So $\frac{c_1\norm{\bv{A}_{\setminus k}}_F^2}{k} \bar \tau_i(\bv{A}) \ge \norm{(\bv{A}-\bv{ZZ}^T\bv{A})_i}_2^2$ for all $i$ and hence $\frac{c_1\norm{\bv{A}_{\setminus k}}_F^2}{k} \tilde \tau_i \ge \norm{(\bv{A}-\bv{ZZ}^T\bv{A})_i}_2^2$. 

$\bv{C}_2$ is a set of $c \log(1/\delta)/\epsilon \cdot \sum_i \tilde \tau_i$ columns sampled with probability proportional to approximate ridge leverage scores. Consider forming $\bv{C}_2'$ by setting $\left (\bv{C}_2\right)_i = \bv{0}$ with probability: $$\frac{\norm{(\bv{A}-\bv{ZZ}^T\bv{A})_{j(i)}}_2^2}{\frac{c_1\norm{\bv{A}_{\setminus k}}_F^2}{k} \tilde \tau_{j(i)}},$$ where $j(i)$ is just the index of the column of $\bv{A}$ that $\left (\bv{C}_2\right )_i$ is equal to. Clearly, if not equal to $\bv{0}$, each column of $\bv{C}_2'$ is equal to $\bv{a}_i$ with probability proportional to the adaptive sampling probability $\norm{(\bv{A}-\bv{ZZ}^T\bv{A})_i}_2^2$. Additionally, in expectation, the number of nonzero columns will be:
\begin{align*}
\left (c \log(1/\delta)/\epsilon \cdot \sum_i \tilde \tau_i\right) \cdot \sum_j \left [ \frac{\tilde \tau_j}{ \sum_i \tilde \tau_i} \frac{\norm{(\bv{A}-\bv{ZZ}^T\bv{A})_{j}}_2^2}{\frac{c_1\norm{\bv{A}_{\setminus k}}_F^2}{k} \tilde \tau_{j}} \right ] = \frac{ck \log(1/\delta)}{c_1\epsilon}\cdot  \frac{\norm{\bv{A}-\bv{Z}\bv{Z}^T\bv{A}}_F^2}{\norm{\bv{A}_{\setminus k}}_F^2}.
\end{align*}

By a Chernoff bound, with probability $1-\delta/2$ at least half this number of columns will be nonzero, and
by Lemma \ref{adaptive_sampling_primative}, for large enough $c$, conditioning on the above column count bound holding, $[\bv{C}_1 \cup \bv{C}_2' ]$ is a $(1+\epsilon)$ error column subset for $\bv{A}$ with probability $1-\delta/2$. Just noting that $\colspan([\bv{C}_1 \cup \bv{C}_2' ]) \subseteq \colspan([\bv{C}_1 \cup \bv{C}_2 ])$ and union bounding over the two possible fail conditions, gives that $[\bv{C}_1 \cup \bv{C}_2 ] = \bv{C}$ is a $(1+\epsilon)$ column subset with probability at least $1-\delta$.

\subsubsection{Stronger Containment Guarantee}
It now remains to show the second condition of Theorem \ref{css_intro}: $\bv{C}$ contains a subset of $O(\sum_i \tilde \tau_i/\epsilon)$ columns that also satisfies Definition \ref{def:css}. This follows from noting that we can apply, for example, the polynomial time deterministic column selection algorithm of \cite{kmeansPaper} to produce a matrix $\bv{C}_1'$ with $O(k)$ columns that is both a constant error additive-multiplicative spectral approximation and a constant error projection-cost preserving sample for $\bv{C}_1$. If $\bv{C}_1'$ has constant error for $\bv{C}_1$, it does for $\bv{A}$ as well and so is a constant error column subset. 

$\bv{C}_2$ contains $O(\log(1/\delta))$ sets of $O(\sum_i \tilde \tau_i/\epsilon)$ columns, $\bv{C}_2^1, \bv{C}_2^2,\ldots,\bv{C}_2^{O(\log(1/\delta))}$. By our argument above, for each $\bv{C}_2^i$, $[\bv{C}_1', \bv{C}_2^i ]$ is a $(1+\epsilon)$ error column subset of $\bv{A}$ with constant probability. So with probability $1-\delta$, at least one $[\bv{C}_1', \bv{C}_2^i ]$ is good. This set contains just $O(k + \sum_i \tilde \tau_i/\epsilon)=O( \sum_i \tilde\tau_i/\epsilon)$ columns, giving the theorem.
\end{proof}

\section{Monotonicity of Ridge Leverage Scores}
\label{sec:mono}

With our main sampling results in place, we focus on the algorithmic problem of how to efficiently approximate the ridge leverage scores of a matrix $\bv{A}$. In the offline setting, we will show that these scores can be approximated in $O(\nnz(\bv{A}))$ time using a recursive sampling algorithm.  We will also show how to compute and sample by the scores in a single-pass column stream. 

Both of these applications will require a unique stability property of the ridge leverage scores:
\begin{lemma}[Ridge Leverage Score Monotonicity]\label{monotonicity}
For any $\bv{A} \in \mathbb{R}^{n \times d}$ and vector $\bv{x} \in \R^{n}$, for every $i \in 1,\ldots,d$ we have:
\begin{align*}
\bar \tau_i(\bv{A}) \le \bar \tau_i(\bv{A} \cup \bv{x}),
\end{align*}
where $\bv{A} \cup \bv{x}$ is simply $\bv{A}$ with $\bv{x}$ appended as its final column.
\end{lemma}

This statement is extremely natural, given that leverage scores are meant to be a measure of importance. It ensures that the importance of a column can only decrease when additional columns are added to $\bv{A}$. While it holds for standard leverage scores, surprisingly no prior low-rank leverage scores satisfy this property.

We begin by defining the \emph{generalized ridge leverage score} as the ridge leverage score of a column estimated using a matrix other than $\bv{A}$ itself.
\begin{definition}[Generalized Ridge Leverage Score]
For any $\bv{A} \in \mathbb{R}^{n \times d}$ and $\bv{M} \in \mathbb{R}^{n \times d'}$,
the $i^{th}$ generalized ridge leverage score of $\bv{A}$ with respect to $\bv{M}$ is defined as:
\begin{align*}
\bar \tau_i^{\bv M}(\bv{A}) = \begin{cases}\bv{a}_i^T \left (\bv{MM}^T + \frac{\norm{\bv{M} - \bv{M}_k}_F^2}{k}\bv{I}\right )^+\bv{a}_i & \text{for } \bv{a}_i \in \colspan\left (\bv{MM}^T + \frac{\norm{\bv{M} - \bv{M}_k}_F^2}{k}\bv{I}\right)\\
\infty & \text{otherwise.}
\end{cases}
\end{align*}
\end{definition}
This definition is the intuitive one. Since our goal is typically to compute over-estimates of $\bar\tau_i(\bv{A})$ using $\bv{M}$, if $\bv{a}_i$ does not fall in the span of $\bv{MM}^T + \frac{\norm{\bv{M} - \bv{M}_k}_F^2}{k}\bv{I}$ we conservatively set its generalized leverage score to $\infty$ instead of $0$. Note that this case only applies when $\bv{M}$ is rank $k$ and thus $\frac{\norm{\bv{M} - \bv{M}_k}_F^2}{k}\bv{I}$ is $0$. 

We now prove a general monotonicity theorem, from which Lemma \ref{monotonicity} follows immediately by setting $\bv{M} = \bv{A}$ and $\bv{A} = \bv{A} \cup \bv{x}$.

\begin{theorem}[Generalized Monotonicity Bound]\label{monotonicity2}
For any $\bv{A} \in \mathbb{R}^{n \times d}$ and $\bv{M} \in \mathbb{R}^{n \times d'}$ with $\bv{M}\bv{M}^T \preceq \bv{A}\bv{A}^T$ we have:
\begin{align*}
\bar \tau_i(\bv{A}) \le \bar \tau_i^{\bv{M}}(\bv{A}).
\end{align*}
\end{theorem}
\begin{proof}
We first note that $\norm{\bv{M} - \bv{M}_k}_F^2 \le \norm{\bv{A}-\bv{A}_k}_F^2$ since, letting $\bv{P}_k$ be the projection onto the top $k$ column singular vectors of $\bv{A}$, by the optimality of $\bv{M}_k$ we have:
\begin{align*}
\norm{\bv{M} - \bv{M}_k}_F^2 \le \norm{(\bv{I}-\bv{P}_k)\bv{M}}_F^2 \le \norm{(\bv{I}-\bv{P}_k)\bv{A}}_F^2 = \norm{\bv{A}-\bv{A}_k}_F^2.
\end{align*}
Accordingly, $$\bv{MM}^T + \frac{\norm{\bv{M}-\bv{M}_k}_F^2}{k}\bv{I} \preceq \bv{AA}^T + \frac{\norm{\bv{A}-\bv{A}_k}_F^2}{k}\bv{I}.$$ Let $\bv{R}$ be a projection matrix onto the column span of $\bv{MM}^T + \frac{\norm{\bv{M} - \bv{M}_k}_F^2}{k} \bv{I}$. Since for any PSD matrices $\bv{B}$ and $\bv{C}$ with the same column span, $\bv{B} \preceq \bv{C}$ implies $\bv{B}^+ \succeq \bv{C}^+$ (see \cite{psdInvert}) we have:
\begin{align*}
\bv{R}\left (\bv{MM}^T + \frac{\norm{\bv{M}-\bv{M}_k}_F^2}{k}\bv{I}\right )^+\bv{R} \succeq \bv{R} \left ( \bv{AA}^T + \frac{\norm{\bv{A}-\bv{A}_k}_F^2}{k}\bv{I}\right )^+\bv{R}.
\end{align*}
For any $\bv{a}_i$ not lying in $\colspan\left (\bv{MM}^T + \frac{\norm{\bv{M} - \bv{M}_k}_F^2}{k} \bv{I} \right) $, $\bar \tau_i^{\bv{M}}(\bv{A}) = \infty$ and the theorem holds trivially. Otherwise, we have $\bv{R}\bv{a}_i = \bv{a}_i$ and so:
\begin{align*}
\bar \tau_i(\bv{A}) = \bv{a}_i^T\bv{R} \left ( \bv{AA}^T + \frac{\norm{\bv{A}-\bv{A}_k}_F^2}{k}\bv{I}\right )^+\bv{R}\bv{a}_i \le \bv{a}_i^T \bv{R}\left (\bv{MM}^T + \frac{\norm{\bv{M}-\bv{M}_k}_F^2}{k}\bv{I}\right )^+\bv{R} \bv{a}_i = \bar \tau_i^{\bv{M}}(\bv A).
\end{align*}
This gives the theorem.
\end{proof}

\section{Recursive Ridge Leverage Score Approximation}
\label{sec:nnza}
With Theorem \ref{monotonicity2} in place, we are ready to prove that ridge leverage scores can be approximated in $O(\nnz(\bv{A}))$ time. Our work closely follows \cite{cohen2015uniform}, which shows how to approximate traditional leverage scores via recursive sampling. 

\subsection{Intuition and Preliminaries}
The central idea behind recursive sampling is as follows: if we uniformly sample, for example, $1/2$ of $\bv{A}$'s columns to form $\bv{C}$ and compute ridge leverage score estimates with respect to just these columns, by monotonicity, the estimates will \emph{upper bound} $\bv{A}$'s true ridge leverage scores. While some of these upper bounds will be crude, we can show that their overall sum is small. 

Accordingly, we can use the estimates to sample $O(k \log k)$ columns from $\bv{A}$ to obtain a constant factor additive-multiplicative spectral approximation by Theorem \ref{matrix_chernoff}, as well as a constant factor projection-cost preserving sample by Theorem \ref{pcp_intro}. This approximation is enough to obtain constant factor estimates of the ridge leverage scores of $\bv{A}$. 

$\bv{C}$ may still be relatively large (e.g. half the size of $\bv{A}$), but it can be recursively approximated via the same sampling scheme, eventually giving our input sparsity time algorithm.

We first give a foundational lemma showing that an approximation of the form given by Theorems \ref{matrix_chernoff} and \ref{pcp_intro} is enough to give constant factor approximations to ridge leverage scores.
\begin{lemma}\label{leverage_approx} Assume that, for an $\epsilon \leq 1/2$, we have $\bv{C}$ satisfying equation \eqref{freq_dirs_two_sided} from Theorem \ref{matrix_chernoff}:
\begin{align*}
(1-\epsilon) \bv{C} \bv{C}^T - \frac{\epsilon}{k} \norm{\bv{A} - \bv{A}_k}_F^2 \bv{I} \preceq \bv{A}\bv{A}^T \preceq (1+\epsilon) \bv{C} \bv{C}^T + \frac{\epsilon}{k} \norm{\bv{A} - \bv{A}_k}_F^2 \bv{I},\end{align*}
along with equation \eqref{eq:pcp} from Definition \ref{def:pcp}:
\begin{align*}
(1-\epsilon)\|\bv{A} - \bv{X}\bv{A}\|^2_F \leq \|\bv{C} - \bv{X}\bv{C}\|^2_F\leq (1+\epsilon)\|\bv{A} - \bv{X}\bv{A}\|^2_F \text{, $\forall$ rank $k$ } \bv{X}.
\end{align*}
Then for all $i$, $$(1-4\epsilon) \bar \tau_i(\bv{A}) \le \bar \tau_i^{\bv C}(\bv{A}) \le (1+4\epsilon) \bar \tau_i(\bv{A}).$$
\end{lemma}
\begin{proof}
Let $\bv{P}_k$ be the projection onto $\bv{A}$'s top $k$ column singular vectors. By the optimality of $\bv{C}_k$ in approximating $\bv{C}$ and the projection-cost preservation condition, we know that $\norm{\bv{C}-\bv{C}_k}_F^2 \le \norm{\bv{C}-\bv{P}_k\bv{C}}_F^2 \le (1+\epsilon)\norm{\bv{A}-\bv{A}_k}_F^2.$ Also, letting $ \bv{\tilde P}_k$ be the projection onto $\bv{C}$'s top $k$ column singular vectors, we have $(1-\epsilon)\norm{\bv{A}-\bv{A}_k}_F^2 \le (1-\epsilon)\norm{\bv{A}-\bv{\tilde P}_k\bv{A}}_F^2 \le \norm{\bv{C}-\bv{C}_k}_F^2$. So overall:
\begin{align}
\label{eq:tail_mult}
(1-\epsilon)\norm{\bv{A}-\bv{A}_k}_F^2 \le \norm{\bv{C}-\bv{C}_k}_F^2 \le (1+\epsilon)\norm{\bv{A}-\bv{A}_k}_F^2.
\end{align}
Using the guarantee from Theorem \ref{matrix_chernoff} we have:
\begin{align*}
(1-\epsilon) \bv{C} \bv{C}^T + \frac{(1-\epsilon)\norm{\bv{A} - \bv{A}_k}_F^2}{k}  \bv{I} &\preceq \bv{AA}^T + \frac{\norm{\bv{A} - \bv{A}_k}_F^2}{k} \bv{I} \preceq (1+\epsilon) \bv{C} \bv{C}^T + \frac{(1+\epsilon)\norm{\bv{A} - \bv{A}_k}_F^2}{k} \bv{I}.
\end{align*}
Combining with our bound on $\norm{\bv{C}-\bv{C}_k}_F^2$ gives:
\begin{align*}
(1-\epsilon) \bv{C} \bv{C}^T + \frac{\frac{(1-\epsilon)}{(1+\epsilon)}\norm{\bv{C} - \bv{C}_k}_F^2}{k}  \bv{I} &\preceq \bv{AA}^T + \frac{\norm{\bv{A} - \bv{A}_k}_F^2}{k} \bv{I} \preceq (1+\epsilon) \bv{C} \bv{C}^T + \frac{\frac{(1+\epsilon)}{(1-\epsilon)}\norm{\bv{C} - \bv{C}_k}_F^2}{k} \bv{I},
\end{align*}
and when $\epsilon \leq 1/2$, we can simplify to:
\begin{align*}
(1-4\epsilon) \left ( \bv{C} \bv{C}^T + \frac{\norm{\bv{C} - \bv{C}_k}_F^2}{k}  \bv{I} \right )&\preceq \bv{AA}^T + \frac{\norm{\bv{A} - \bv{A}_k}_F^2}{k} \bv{I} \preceq (1+4\epsilon) \left ( \bv{C} \bv{C}^T + \frac{\norm{\bv{C} - \bv{C}_k}_F^2}{k} \bv{I} \right ).
\end{align*}

If $\norm{\bv{A} - \bv{A}_k}_F^2 = 0$, and thus by \eqref{eq:tail_mult} $\norm{\bv{C} - \bv{C}_k}_F^2 = 0$, then $\bv{A}$ and $\bv{C}$ must have the same column span or else it could not hold that $(1-4\epsilon)\bv{CC}^T\preceq\bv{A}\bv{A}^T\preceq(1+4\epsilon)\bv{CC}^T$. On the other hand, if $\norm{\bv{A} - \bv{A}_k}_F^2 > 0$, and thus by \eqref{eq:tail_mult} $\norm{\bv{C} - \bv{C}_k}_F^2 > 0$, both $\bv{AA}^T + \frac{\norm{\bv{A} - \bv{A}_k}_F^2}{k}\bv{I}$ and $ \bv{C} \bv{C}^T + \frac{\norm{\bv{C} - \bv{C}_k}_F^2}{k}\bv{I}$ span all of $\mathbb{R}^n$. Either way, the two matrices have the same span and so by \cite{psdInvert} we have:
\begin{align*}
(1-4\epsilon) \left (\bv{AA}^T + \frac{\norm{\bv{A}_{\setminus k}}_F^2}{k} \bv{I} \right )^+ \preceq \left ( \bv{C} \bv{C}^T + \frac{\norm{\bv{C}_{\setminus k}}_F^2}{k}  \bv{I} \right )^+ \preceq (1+4\epsilon) \left (\bv{AA}^T + \frac{\norm{\bv{A}_{\setminus k}}_F^2}{k} \bv{I} \right )^+,
\end{align*}
which gives the lemma.
\end{proof}

Our next lemma, which is analogous to Theorem 2 of \cite{cohen2015uniform}, shows that by reweighting a 
small number of columns in $\bv{A}$, we can obtain a matrix with all ridge leverage scores bounded by a small constant, which ensures that it can be well approximated by uniform sampling.

\begin{lemma}[Ridge Leverage Score Bounding Column Reweighting]\label{weighting_existance}
For any $\bv{A} \in \R^{n \times d}$ and any score upper bound $u>0$, there exists a diagonal matrix $\bv{W} \in \mathbb{R}^{d\times d}$ with $\bv{0} \preceq \bv{W} \preceq \bv{I}$ such that:
\begin{align}\label{levUpper}
\forall i, \bar \tau_i\left(\bv{AW}\right) \le u,
\end{align}
and
\begin{align}\label{levSum}
|\{i : \bv{W}_{ii} \neq 1 \}| \le \frac{3k}{u}.
\end{align}
\end{lemma}  
\begin{proof}
This result follows from Theorem 2 of \cite{cohen2015uniform}, to which we refer the reader for details.
To show the existence of a reweighting $\bv{W}$ satisfying \eqref{levUpper} and \eqref{levSum}, we will argue that a simple iterative process (which we never actually need to implement) converges on the necessary reweighting.

Specifically, if a column has too high of a leverage score, we simply decrease its weight until $\bar \tau_i(\bv{AW}) \le u$. We want to argue that, given $\bv{AW}_0$ with $\bar \tau_i(\bv{AW}_0) > u$, we can decrease the weight on $\bv{a}_i$ to produce $\bv{W}_1$ with $\bar \tau_i(\bv{AW}_1) \le u$. By Lemma 5 of \cite{cohen2015uniform} we can always decrease the weight on $\bv{a}_i$ to ensure $\tau_i(\bv{AW}_1) \le u$, where $\tau_i(\cdot)$ is the traditional leverage score. And since $\left ( \bv{A}\bv{W}_1^2\bv{A}^T + \frac{\norm{(\bv{AW}_1)_{\setminus k}}_F^2}{k} \bv{I} \right )^+ \preceq \left (\bv{A}\bv{W}_1^2\bv{A}^T \right )^+ $, $\bar \tau_i(\bv{AW}_1) \le \tau_i(\bv{AW}_1)$, so an equivalent or smaller weight decrease suffices to decrease $\bar \tau_i(\bv{AW}_1)$ below $u$.

Furthermore, we can see that $\bar \tau_i(\bv{AW})$ is continuous with respect to $\bv{W}$. This is due to the fact that both the traditional leverage scores of $\bv{AW}$ (shown in Lemma 6 of \cite{cohen2015uniform}) and $\norm{(\bv{AW})_{\setminus k}}_F^2$ are continuous in $\bv{W}$.  From Theorem 2 of \cite{cohen2015uniform}, continuity implies that iteratively reweighting individual columns converges, and thus there is always exists a reweighting satisfying \eqref{levUpper}. 

It remains to show that this reweighting satisfies \eqref{levSum}. By continuity, we can always decrease $\bar \tau_i(\bv{AW}_0)$ to exactly $u$ unless $\bar \tau_i(\bv{AW}) = 1$, in which case the only option is to set the weight on the column to $0$ and hence set $\bar \tau_i(\bv{AW}) = 0$. However, if $\norm{\bv{A}_{\setminus k}}_F^2 > 0$, then \emph{every} ridge leverage score is strictly less than $1$. If $\norm{\bv{A}_{\setminus k}}_F^2 =0$, then $\bv{A}$ has rank $k$, the ridge leverage scores are the same as the true leverage scores, and the number of columns with leverage score $1$ is at most $k$. Therefore, by Theorem 2 of \cite{cohen2015uniform}, monotonicity, and the fact that $\sum_i \bar \tau_i(\bv{AW}) \le 2k$ for any $\bv{W}$, we have the lemma.
\end{proof}

\subsection{Uniform Sampling for Ridge Leverage Score Approximation}

Using Lemmas \ref{leverage_approx} and \ref{weighting_existance} we can prove the key step of our recursive sampling method: if we uniformly sample columns from $\bv{A}$ and use them to estimate ridge leverage scores, these scores can be used to resample a set of columns that give constant factor ridge leverage scores approximations.

\begin{theorem}[Ridge Leverage Score Approximation via Uniform Sampling]\label{uniform_stronger}
Given $\bv{A} \in \mathbb{R}^{n \times d}$, construct $\bv{C}_u$ by independently sampling each column of $\bv{A}$ with probability $\frac{1}{2}$. Let 
\begin{align*}
\tilde \tau_i = \min \left \{1, \bar \tau_i^{\bv{C}_u}(\bv A) \right \}.
\end{align*}
If we form $\bv{C}$ by sampling each column of $\bv{A}$ independently with probability $p_i = \min \left \{1, \tilde \tau_i c_1\log(k/\delta)\right \}$ and reweighting by $1/\sqrt{p_i}$ if selected, then for large enough constant $c_1$, with probability $1-\delta$, $\bv{C}$ will have just $O(k\log(k/\delta))$ columns and will satisfy the conditions of Lemma \ref{leverage_approx} for some constant error. Accordingly, we have:
\begin{align*}
\frac{1}{2}  \bar \tau_i(\bv{A}) \le \bar \tau_i^{\bv{C}}(\bv{A}) \le 2 \bar \tau_i(\bv{A}).
\end{align*}
\end{theorem}
\begin{proof}
Clearly $\bv{C}_u \bv{C}_u^T \preceq \bv{A} \bv{A}^T$, so by the monotonicity shown in Theorem \ref{monotonicity2} we have $\bar \tau_i^{\bv{C}_u}(\bv A) \ge \bar \tau_i(\bv A)$. Since $\bar \tau_i(\bv A)$ is always $\le 1$, it follows that $\tilde \tau_i = \min \left \{1,\bar \tau_i^{\bv{C}_u}(\bv A) \right \} \ge \bar \tau_i(\bv A)$. 
Then we can just use the $\tilde \tau_i$'s obtained from $\bv{C}_u$ in independent sampling versions of Theorems \ref{matrix_chernoff} and \ref{pcp_intro}, which can be proven from Lemmas \ref{matrix_chernoff2} and \ref{trace:bound2} in Appendix \ref{app:without_replacement}. Accordingly, with probability $1-\delta/3$, $\bv{C}$ gives a constant factor additive-multiplicative spectral approximation and projection-cost preserving sample of $\bv{A}$. Hence by Lemma \ref{leverage_approx}, $\bar \tau_i^{\bv{C}}(\bv{A})$ is a constant factor approximation to $\bar \tau_i(\bv{A})$.

To prove the theorem, we still have to show that $\bv{C}$ does not have too many columns. Its expected number of columns is: 
\begin{align*}
\sum_i p_i = \sum_i \min \left \{1, \tilde \tau_i c_1\log(k/\delta)\right \}.
\end{align*} 
By Lemma \ref{weighting_existance} instantiated with $u = \frac{1}{2c_2\log(k/\delta)}$, we know that there is some reweighting matrix $\bv{W}$ with only $3k\cdot 2c_2\log(k/\delta)$ entries not equal to $1$ such that $\bar \tau_i(\bv{AW}) \le  \frac{1}{2c_2\log(k/\delta)}$ for all $i$. We have:
\begin{align}
\sum_i p_i &= \sum_{i:\bv{W}_{ii} \neq 1} p_i + \sum_{i: \bv{W}_{ii} = 1} p_i \nonumber\\
&\le 6kc_2\log(k/\delta)  + \sum_{i: \bv{W}_{ii} = 1} c\log(k/\delta) \cdot \bar \tau_i^{\bv{C}_u}(\bv{A})\nonumber\\
&=  6kc_2\log(k/\delta) + c_1\log(k/\delta)\cdot \sum_{i: \bv{W}_{ii} = 1} \bar \tau_i^{\bv{C}_u}(\bv{AW})\nonumber\\
&\le 6kc_2\log(k/\delta)  + c_1\log(k/\delta)\cdot \sum_{i: \bv{W}_{ii} = 1} \ \bar \tau_i^{\bv{C}_u\bv{W}}(\bv{AW})\nonumber\\
&\le 6kc_2\log(k/\delta)  + c_1\log(k/\delta)\cdot \sum_{i} \ \bar \tau_i^{\bv{C}_u\bv{W}}(\bv{AW}). \label{eq:final_sum_bound}
\end{align}
Now, since every ridge leverage score of $\bv{AW}$ is bounded by $\frac{1}{2c_2\log(k/\delta)}$, if $c_2$ is set large enough, the uniformly sampled $\bv{C}_u\bv{W}$ is a proper ridge leverage score oversampling of $\bv{AW}$, except that its columns were not reweighted by a factor of $2$ (they were each sampled with probability $1/2$). 

Accordingly, with probability $1-\delta/3$, $2\bv{C}_u\bv{W}$ satisfies the approximation conditions of Lemma \ref{leverage_approx}
for $\bv{AW}$ with $\epsilon = 1/2$. Thus, for all $i$, $\frac{1}{2}\bar \tau_i^{\bv{C}_u\bv{W}}(\bv{AW}) = \bar \tau_i^{2\bv{C}_u\bv{W}}(\bv{AW})\le 3 \bar \tau_i(\bv{AW}).$ By Lemma \ref{sumScores}, $\sum_i \bar\tau_i(\bv{AW}) \le 2k$ so overall $\sum_{i} \ \bar \tau_i^{\bv{C}_u\bv{W}}(\bv{AW}) \leq 12k$. Plugging back in to \eqref{eq:final_sum_bound}, we conclude that $\bv{C}$ has
 $O(k\log(k/\delta))$ columns in expectation, and actually with probability $1-\delta/3$ by a Chernoff bound.
Union bounding over our failure probabilities gives the theorem.
\end{proof}

\subsection{Basic Recursive Algorithm}

Theorem \ref{uniform_stronger} immediately proves correct Algorithm \ref{halving} for ridge leverage score approximation:

\begin{algorithm}[H]
\caption{\algoname{Repeated Halving}}
{\bf input}: $\bv{A} \in \mathbb{R}^{n \times d}$\\
{\bf output}: A reweighted column sample $\bv{C} \in \mathbb{R}^{n \times O(k\log (k/\delta))}$ satisfying the guarantees of Theorems \ref{matrix_chernoff} and \ref{pcp_intro} with constant error.
\begin{algorithmic}[1]
\State{Uniformly sample $\frac{d}{2}$ columns of $\bv{A}$ to form $\bv{C}_u$}
\State{If $\bv{C}_u$ has $>O(k \log k )$ columns, \textbf{recursively} apply \textsc{Repeated Halving} to compute a constant factor approximation $\bv {\tilde C}_u$ for $\bv{ C}_u$ with $O(k \log k )$ columns.}
\State{Compute generalized ridge leverage scores of $\bv{A}$
with respect to $\bv{\tilde C}_u$
}
\State{Use these estimates to sample columns of $\bv{A}$ to form $\bv{C}$}\\
\Return{$\bv{C}$}
\end{algorithmic}
\label{halving}
\end{algorithm}

Note that, by Lemma \ref{leverage_approx}, generalized ridge leverage scores computed with respect to $ \bv{\tilde C}_u$ are constant factor approximations to generalized ridge leverage scores computed with respect to $ \bv{C}_u$. Accordingly, by Theorem \ref{uniform_stronger}, we conclude that $\bv{C}$ is a valid ridge leverage score sampling of $\bv{A}$.

Before giving our full input sparsity time result, we warm up with a simpler theorem that obtains a slightly suboptimal runtime.

\begin{lemma}
\label{lem:algo1}
A simple implementation of Algorithm \ref{halving} that succeeds with probability $1-\delta$ runs in $O\left ( \nnz(\bv{A})\log(d/\delta)\right) + \tilde O(nk^2)$ time.
\end{lemma}
\noindent For clarity of exposition, we use $\tilde O(\cdot)$ to hide log factors in $k$, $d$, and $1/\delta$ on the lower order term.
\begin{proof}
The algorithm has $\log(d/k)$ levels of recursion and, since we sample our matrix uniformly, $\nnz(\bv{A})$ is cut approximately in half at each level, with high probability. It thus suffices to show that the work
done at the top level is $O\left ( \nnz(\bv{A})\log (d/\delta)\right ) + \tilde O(nk^2)$. 

To compute the generalized ridge leverage scores of $\bv{A}$ with respect to $\bv{\tilde C}_u$ we must (approximately) compute, for each $\bv{a}_i$,
\begin{align}
\label{eq:what_we_need_to_eval}
\bv{a}_i^T \left (\bv{\tilde C}_u \bv{\tilde C}_u^T + \frac{\norm{\bv{\tilde C}_u- (\bv{\tilde C}_u)_k}_F^2}{k} \bv{I} \right )^+ \bv{a}_i.
\end{align}
We are going to ignore that $\left (\bv{\tilde C}_u \bv{\tilde C}_u^T + \frac{\norm{\bv{\tilde C}_u- (\bv{\tilde C}_u)_k}_F^2}{k} \bv{I}\right )$ could be sparse and well conditioned (and thus ideal for iterative solvers) and use direct methods for simplicity.

Let $\lambda$ denote $\frac{\norm{\bv{\tilde C}_u- (\bv{\tilde C}_u)_k}_F^2}{k}$ and let $\bv{R}\in \R^{n\times \tilde{O}(k)}$ be an orthonormal basis containing the left singular vectors of  $\bv{\tilde C}_u$. We can rewrite:
\begin{align*}
 \left (\bv{\tilde C}_u \bv{\tilde C}_u^T + \lambda \bv{I} \right) = \bv{\tilde C}_u \bv{\tilde C}_u^T + \lambda\bv{R}\bv{R}^T  + \lambda\left(\bv{I}-\bv{R}\bv{R}^T\right),
\end{align*}
and accordingly, using the fact that $\bv{R}\bv{R}^T$ and $(\bv{I}-\bv{R}\bv{R}^T)$ are orthogonal,
\begin{align*}
\left (\bv{\tilde C}_u \bv{\tilde C}_u^T + \lambda \bv{I} \right)^+ = \left(\bv{\tilde C}_u \bv{\tilde C}_u^T + \lambda\bv{R}\bv{R}^T\right)^+ + \frac{1}{\lambda}\left(\bv{I}-\bv{R}\bv{R}^T\right).
\end{align*}
Now, using an SVD of $\bv{\tilde C}_u$, which can be computed in $\tilde O(nk^2)$ time, we compute $\lambda$ and then write  $\left(\bv{\tilde C}_u \bv{\tilde C}_u^T + \lambda\bv{R}\bv{R}^T\right)^+$ as $\bv{R}\bv{\Sigma}^{-2}\bv{R}^T$ for some diagonal matrix $\bv{\Sigma} \in \R^{\tilde{O}(k)\times\tilde{O}(k)}$. Accordingly, to evaluate \eqref{eq:what_we_need_to_eval}, we need just need to compute:
\begin{align*}
\bv{a}_i^T\left(\bv{R}\bv{\Sigma}^{-2}\bv{R}^T + \frac{1}{\lambda}\left(\bv{I}-\bv{R}\bv{R}^T\right)\right)\bv{a}_i = \|\left(\bv{R}^T\bv{\Sigma}^{-1}\bv{R}^T + \frac{1}{\sqrt{\lambda}}\left(\bv{I}-\bv{R}\bv{R}^T\right)\right)\bv{a}_i\|_2^2.
\end{align*}
Since $\bv{R}$ has $\tilde{O}(k)$ columns, naively evaluating this norm for all of $\bv{A}$'s columns would require a total of $\tilde{O}(\nnz(\bv{A})k)$ time.
However, we can accelerate the computation via a Johnson-Lindenstrauss embedding technique that has become standard for computing regular leverage scores \cite{Spielman:2008}. 

 Specifically, denoting $\left(\bv{R}^T\bv{\Sigma}^{-1}\bv{R}^T + \frac{1}{\sqrt{\lambda}}\left(\bv{I}-\bv{R}\bv{R}^T\right)\right)$ as $\bv{M}$, we can embed $\bv{M}$'s columns into $O(\log (d/\delta))\times n$ dimensions by multiplying on the left by a matrix $\bv{\Pi} \in \mathbb{R}^{O(\log (d/\delta) \times n}$ with scaled random Gaussian or random sign entries.
Even though $\bv{M}$ is $n\times n$, we can perform the multiplication in $\tilde{O}(nk\log(d/\delta))$ by working with our factored form of the matrix.

By standard Johnson-Lindenstrauss results, 
$\norm{ \bv{\Pi} \bv{M} \bv{a}_i }_2^2$ will be within a constant factor of $\norm{ \bv{M} \bv{a}_i }_2^2$ for all $i$ with probability $1-\delta$. Furthermore, we can evaluate $\norm{ \bv{\Pi} \bv{M} \bv{a}_i }_2^2$ for all $\bv{a}_i$ in $O\left (\nnz(\bv{A})\log (d/\delta)\right )$ total time. Our final cost for approximating all ridge leverage scores is thus $O\left (\nnz(\bv{A})\log (d/\delta)\right ) + \tilde O\left (nk^2\right)$ time, which gives the lemma.
\end{proof}

\subsection{True Input-Sparsity Time}
Sharpening Lemma \ref{lem:algo1} to eliminate log factors on the $\nnz(\bv{A})$ runtime term requires standard optimizations for approximating leverage scores with respect to a subsample \cite{pengV2,cohen2015uniform}.

In particular, we can actually apply a Johnson-Lindenstrauss embedding matrix to $\bv{M}$ with just $\theta^{-1}$ rows for some small constant $\theta$. 
Doing so will approximate each ridge leverage score to within a factor of $d^{\theta}$ with high probability (see Lemma 4.5 of \cite{pengV2} for example). 

This level of approximation is sufficient to resample $O\left (kd^{\theta}\log(k/\delta)\right )$ columns from $\bv{A}$ to form an approximation $\bv{C}'$ that satisfies the guarantees of Theorems \ref{matrix_chernoff} and \ref{pcp_intro}. To form $\bv{C}$, we further sample $\bv{C}'$ down to $O(k\log (k/\delta))$ columns using its ridge leverage scores, which takes $\tilde O(nk^2d^{2\theta})$ time. Finally, under the reasonable assumption that $\epsilon$ and $\delta$ are $\poly(n)$, we can also assume $d = \poly(n)$. Otherwise, $\nnz(\bv A) \ge d$ dominates the $\tilde O(nk^2d^{2\theta})$ term. This yields the following: 
\begin{lemma}\label{inputSparsitySampling}
\label{lem:algo2}
An optimized implementation of Algorithm \ref{halving}, succeeding with probability $1-\delta$, runs in time $O\left ( \theta^{-1} \nnz(\bv{A}) \right ) + \tilde O(n^{1+\theta}k^{2})$ time, for any $\theta \in (0,1]$.
\end{lemma}

Once we have used Algorithm \ref{halving} to obtain $\bv{C}$ satisfying the guarantees of Theorems \ref{matrix_chernoff} and \ref{pcp_intro} with constant error, we can approximate $\bv{A}$'s ridge leverage scores and resample one final time to obtain an $\epsilon$ error projection-cost preserving sketch. This immediately yields our main algorithmic result:

\begin{reptheorem}{nnza_intro} For any $\theta \in (0,1]$, there exists an iterative column sampling algorithm that, in time $O\left (\theta^{-1}\nnz(\bv{A})\right ) + \tilde O \left (\frac{n^{1+\theta}k^{2}}{\epsilon^4}\right)$, returns $\bv{Z} \in \mathbb{R}^{n\times k}$ satisfying:
\begin{align}\label{low_rank_gurantee}
\norm{\bv{A}-\bv{ZZ}^T \bv{A}}_F^2 \le (1+\epsilon) \norm{\bv A - \bv{A}_k}_F^2.
\end{align}
All significant linear algebraic operations of the algorithm involve matrices whose columns are subsets of those of $\bv{A}$, and thus inherit any structure from the original matrix, including sparsity.
\end{reptheorem}
\begin{proof}
We use the same technique as Lemma \ref{inputSparsitySampling}, but in the last round of sampling we select $O\left (\frac{kn^{\theta/2}\log(k/\delta)}{\epsilon^2}\right)$ columns to obtain an $O(\epsilon)$ factor projection-cost preserving sample, $\bv{C}$. Setting $\bv{Z}$ to the top $k$ column singular vectors of $\bv{C}$, which takes $\tilde O(n^{1+\theta}k^{2}/\epsilon^4)$ time, gives \eqref{low_rank_gurantee} \cite{kmeansPaper}.
\end{proof}

\section{Streaming Ridge Leverage Score Sampling}
\label{sec:streaming}

%

We conclude with an application of our results to novel low-rank sampling algorithms for single-pass column streams. While random projection algorithms work naturally in the streaming setting, the study of single-pass streaming column sampling has been limited to the ``full-rank'' case \cite{kelner2011spectral,camMichaelOnlineSparsification,kapralov2014single}. Column subset selection algorithms based on simple norm sampling \emph{are} adaptable to streams, but 
do not give relative error approximation guarantees \cite{Drineas:2006:FMC2,Frieze:2004}.

Relative error algorithms are obtainable by combining our projection-cost preserving sampling procedures with the ``merge-and-reduce'' framework for coresets \cite{saxeBentley,Agarwal:2004,Har-Peled:2004}. 
This approach relies on the composability of projection-cost preserving samples: a $(1+\epsilon)$ error sample for $\bv{A}$ unioned with a $(1+\epsilon)$ error sample for $\bv{B}$ gives a $(1+\epsilon)$ error sample for $[\bv{A},\bv{B}]$ \cite{feldman2013turning}. However, merge-and-reduce requires storage of $O(\log^4dk/\epsilon^2)$ scaled columns from $\bv{A}$, where $d$ is the \emph{length} of our stream (and its value is known ahead of time).


Our algorithms eliminate the $\log^c d$ stream length dependence, storing a fixed number of columns that only depends on $\epsilon$ and $k$. We note that our space bounds are given in terms of the number of real numbers stored. We do not bound the required precision of these numbers, which would include at least a single logarithmic dependence on $d$. In particular, we employ a Frequent Directions sketch that requires words with at least $\Theta(\log(nd))$ bits of precision. Rigorously bounding maximum word-size required for Frequent Directions and our algorithms could be an interesting direction for future work. 

\subsection{General Approach}
The basic idea behind our algorithms is quite simple and follows intuition from prior work on standard leverage score sampling \cite{kelner2011spectral}. Suppose we have some space budget $t$ for storing a column sample $\bv{C}$. As soon as we have streamed in $t$ columns, we can downsample by ridge leverage scores to say $t/2$ columns. As more columns are received, we will eventually reach our storage limit and need to downsample columns again. Doing so naively would compound error: if we resampled $r$ times, our final sample would have error $(1+\epsilon)^r$. 

However, we can avoid compounding error by exploiting Lemma \ref{monotonicity}, which ensures that, as new columns are added, the ridge leverage scores of columns already seen only decrease. Whenever we add a column to $\bv{C}$, we can record the probability it was kept with. In the next round of sampling, we only discard that column with probability equal to the proportion that its ridge leverage score \emph{decreased} by (or keep the column with probability $1$ if the score remained constant).
New columns are simply sampled by ridge leverage score. This process ensures that, at any point in the stream, we have a set of columns sampled by true ridge scores with respect to the matrix seen so far. Accordingly, we will have a $(1+\epsilon)$ error column subset or projection-cost preserving sample at the end of the stream.

This overview hides a number of details, the most important of which is how to compute ridge leverage scores at any given point in the stream with respect to the columns of $\bv{A}$ observed so far. We do not have direct access to these columns since we have only stored a subset of them. We could use the fact that our current sample is projection-cost preserving and can be used to approximate ridge leverage scores (see Lemma \ref{leverage_approx}). However, this approach would introduce sampling dependencies between columns and would require a logarithmic dependence on stream length to ensure that our approximation does not fail at any round of sampling. 

\subsection{Frequent Directions for Approximating Ridge Leverage Scores}
Instead, we use a constant error \emph{deterministic} ``Frequent Directions'' sketch to estimate ridge leverage scores. Introduced in \cite{liberty2013simple} and further analyzed in a series of papers culminating with \cite{ghashami2015frequent}, Frequent Directions sketches are easily maintained in a single-pass column stream of $\bv{A}$. The sketch always provides an approximation $\bv{B} \in \R^{n\times (\ell+1)k}$ guaranteeing:
\begin{align}
\label{eq:true_freq_dirs_guar}
\bv{B}\bv{B}^T \preceq \bv{A}\bv{A}^T \preceq \bv{B}\bv{B}^T + \frac{1}{\ell}\frac{\|\bv{A}_{\setminus k}\|_F^2}{k}.
\end{align}
$\bv{B}$ does not contain columns from $\bv{A}$, so it could be dense even for a sparse input matrix. However, we will only be setting $\ell$ to a small constant. Precise information about $\bv{A}$ will be stored in our column sample $\bv{C}$, which maintains sparsity.

We first show that $\bv{B}$ can be used to compute constant factor approximations to the ridge leverage scores of $\bv{A}$.
\begin{lemma}\label{freq_dirs:lev_approx}
For every column $\bv{a}_i \in \bv{A}$, define
\begin{align*}
\tilde{\tau}_i \eqdef \bv{a}_i^T\left(\bv{B}\bv{B}^T + \frac{\norm{\bv{A}}_F^2 - \norm{\bv{B}_k}_F^2}{k}\bv{I}\right)^+\bv{a}_i.
\end{align*}
If $\bv{B}\in \R^{n\times 3k}$ is a Frequent Directions sketch for $\bv{A}$ with accuracy parameter $\ell = 2$, then
\begin{align*}
\frac{1}{2}\bar{\tau}_i(\bv{A}) \leq \tilde{\tau}_i \leq 2 \bar{\tau}_i(\bv{A}).
\end{align*}
\end{lemma}
$\|\bv{A}\|_F^2$ is obviously computable in a single-pass column stream, so $\tilde{\tau}_i$ can be evaluated in the streaming setting as long as we have access to $\bv{a}_i$.
 \begin{proof}
By the Frequent Directions guarantee, either $\bv{BB}^T = \bv{AA}^T$ giving the lemma trivially, or $\norm{\bv{A}_{\setminus k}}_F^2 \ge 0$. In this case, since $\bv{BB}^T \preceq \bv{AA}^T$, $\norm{\bv{A}}_F^2 - \norm{\bv{B}_k}_F^2 > 0$. So both $\bv{AA}^T + \frac{\norm{\bv{A}_{\setminus k}}_F^2}{k} \bv{I} $ and $\bv{B}\bv{B}^T + \frac{\norm{\bv{A}}_F^2 - \norm{\bv{B}_k}_F^2}{k}\bv{I}$ span all of $\R^n$. Recalling that $\bar \tau_i(\bv{A}) = \bv{a}_i^T \left (\bv{AA}^T + \frac{\norm{\bv{A}_{\setminus k}}_F^2}{k} \bv{I} \right)^+ \bv{a}_i$, to prove the lemma it suffices to show:
\begin{align}
\label{eq:freq_dirs_psd_bound}
\frac{1}{2}\left(\bv{B}\bv{B}^T + \frac{\norm{\bv{A}}_F^2- \norm{\bv{B}_k}_F^2}{k}\bv{I}\right) \preceq \bv{AA}^T + \frac{\norm{\bv{A}_{\setminus k}}_F^2}{k} \bv{I} \preceq 2\left(\bv{B}\bv{B}^T + \frac{\norm{\bv{A}}_F^2 - \norm{\bv{B}_k}_F^2}{k}\bv{I}\right).
\end{align}

Recall that the squared Frobenius norm  of a matrix is equal to the sum of its squared singular values. Additionally, 
a standard property of the relation $\bv{M}\preceq \bv{N}$ is that, for all $i$, the $i^\text{th}$ singular value $ \sigma_i(\bv{M}) \leq \sigma_i(\bv{N})$. From the right hand side of \eqref{eq:true_freq_dirs_guar} it follows that, when $\ell = 2$, $\sigma_i^2(\bv{B}) \geq \sigma_i^2(\bv{A}) - \frac{\|\bv{A}_{\setminus k}\|_F^2}{2k}$. Accordingly, since $\norm{\bv{B}_k}_F^2$ is the sum of the top $k$ singular values of $\bv{B}$,
\begin{align*}
\norm{\bv{A}}_F^2 -\norm{\bv{B}_k}_F^2 \leq  \norm{\bv{A}_{\setminus k}}_F^2 + k\frac{\|\bv{A}_{\setminus k}\|_F^2}{2k} \leq 1.5\|\bv{A}_{\setminus k}\|_F^2.
\end{align*}
Since $\bv{B}\bv{B}^T \preceq \bv{A}\bv{A}^T$, it follows that 
 that $\left(\bv{B}\bv{B}^T + \frac{\norm{\bv{A}}_F^2 - \norm{\bv{B}_k}_F^2}{k}\bv{I}\right) \preceq \bv{AA}^T + 1.5\frac{\norm{\bv{A}_{\setminus k}}_F^2}{k} \bv{I}$, which is more than tight enough to give the left hand side of \eqref{eq:freq_dirs_psd_bound}.

Furthermore , $\norm{\bv{A}}_F^2 - \norm{\bv{B}_k}_F^2 \geq \norm{\bv{A}_{\setminus k}}_F^2$, and since $\ell = 2$, $\bv{B}\bv{B}^T \succeq \bv{A}\bv{A}^T - \frac{\|\bv{A}_{\setminus k}\|_F^2}{2k}$. Overall,
\begin{align*}
\left(\bv{B}\bv{B}^T + \frac{\norm{\bv{A}}_F^2 - \norm{\bv{B}_k}_F^2}{k}\bv{I}\right) \succeq \bv{A}\bv{A}^T + \frac{\|\bv{A}_{\setminus k}\|_F^2}{2k},
\end{align*}
which is more than tight enough to give the right hand side of \eqref{eq:freq_dirs_psd_bound}.
\end{proof}

\subsection{Streaming Column Subset Selection}
Lemma \ref{freq_dirs:lev_approx} gives rise to a number of natural algorithms for rejection sampling by ridge leverage score. The simplest approach is to emulate sampling columns from $\bv{A}$ independently without replacement (see Lemmas \ref{matrix_chernoff2} and \ref{trace:bound2}). However, since sampling without replacement produces a variable number of samples, this method would require a $\log d$ dependence to ensure that our space remains bounded throughout the algorithm's execution with high probability.

Instead, we apply our ``with replacement'' bounds, which sample a fixed number of columns, $t$. We start by describing Algorithm \ref{algo:css} for column subset selection. The constant $c$ used below is the necessary oversampling parameter from Theorem \ref{css_intro}. $\bv{C}\in \R^{n\times t}$ stores our actual column subset and $\bv{D}\in \R^{n\times t}$ stores a queue of new columns. $\bv{B}$ is a Frequent Directions sketch with parameter $\ell = 2$.

\begin{breakablealgorithm}
\caption{\algoname{Streaming Column Subset}}
{\bf input}: $\bv{A} \in \mathbb{R}^{n \times d}$, accuracy $\epsilon$, success probability $(1-\delta)$ \\
{\bf output}: $\bv{C} \in \mathbb{R}^{n \times t}$ such that $t = 32c(k\log k + k\log(1/\delta)/\epsilon)$ and each column $\bv{c}_i$ is equal to column $\bv{a}_j$ with probability $p_j \in \left[\frac{1}{2}\frac{\tilde \tau_j c(k\log k + k\log(1/\delta)/\epsilon)}{t}, \frac{\tilde \tau_j c(k\log k + k\log(1/\delta)/\epsilon)}{t}\right]$ and $\bv{0}$ otherwise, where $\tilde \tau_j \ge 2\bar \tau_j(\bv{A})$ for all $j$ and $\sum_{j=1}^n \tilde \tau_j \leq 16k$. 
\begin{algorithmic}[1]\label{streaming_algo}
\State $count := 1$, $\bv{C} := \bv{0}_{n \times t}$, $\bv{D}:= \bv{0}_{n \times t}$, $frobA:=0$ \Comment{\textcolor{blue}{Initialize storage}}
\State $[\tilde \tau^{old}_{1},...,\tilde \tau^{old}_t] := 1$
\Comment{\textcolor{blue}{Initialize sampling probabilities}}
\For{$i := 1,\ldots, d$} \Comment{\textcolor{blue}{Process column stream}}
\State $\bv{B} := \fd(\bv{B},\bv{a}_i)$
\If{$count \le t$} \Comment{\textcolor{blue}{Collect $t$ new columns}}
\State $\bv{d}_{count} := \bv{a}_i$.
\State $frobA := frobA + \norm{\bv{a}_i}_2^2$ \Comment{\textcolor{blue}{Update $\norm{\bv{A}}_F^2$}}
\State{$count := count + 1$}
\Else \Comment{\textcolor{blue}{Prune columns}}
\State $[\tilde \tau_{1},...,\tilde \tau_t] := \min \left \{[\tilde \tau^{old}_{1},...,\tilde \tau^{old}_t],\re(\bv{B},\bv{C}, frobA)\right \}$
\State $[\tilde \tau^\bv{D}_{1},...,\tilde \tau^\bv{D}_{t}] := \re(\bv{B},\bv{D}, frobA)$
\For {$j := 1,\ldots, t$} 
\If{$\bv{c}_j \neq \bv{0}$} \Comment{\textcolor{blue}{Rejection sample}}
\State With probability $\left(1 - \tilde \tau_{j}/\tilde \tau^{old}_{j}\right)$ set $\bv{c}_j := \bv{0}$ and set $\tilde \tau^{old}_{j} := 1$. 
\State Otherwise set $\tilde \tau^{old}_{j} := \tilde \tau_{j}$.
\EndIf
\If{$\bv{c}_j = \bv{0}$} \Comment{\textcolor{blue}{Sample from new columns in $\bv{D}$}}
\For{$\ell := 1,\ldots, t$}
\State With probability $\frac{\tilde \tau_\ell c(k\log k + k\log(1/\delta)/\epsilon)}{t}$ set $\bv{c}_j := \bv{d}_{\ell}$ and set $\tilde \tau^{old}_{j} := \tilde \tau_{\ell}$
\EndFor
\EndIf
\EndFor
\State $count := 0$
\EndIf
\EndFor
\end{algorithmic}
\begin{algorithmic}[1]
\Function{$\re$}{$\bv{B}$, $\bv{M}\in \R^{n\times t}$, $frobA$} 
\For {$i := t+1,\ldots, d$}
\State $\tilde \tau_k$:= $4\bv{m}_i^T\left(\bv{B}\bv{B}^T + \frac{frobA - \norm{\bv{B}_k}_F^2}{k}\bv{I}\right)^+\bv{m}_i$
\EndFor
\State \textbf{return} $[\tilde \tau_{1},...,\tilde \tau_t]$
\EndFunction
\end{algorithmic}
\label{algo:css}
\end{breakablealgorithm}

To prove the correctness of Algorithm \ref{algo:css}, we first note that, if our output $\bv{C}$ has columns belonging to the claimed distribution, then with probability $(1-\delta)$, $\bv{C}$ is a $(1+\epsilon)$ error column subset for $\bv{A}$ satisfying the guarantees of Theorem \ref{css_intro}. Our procedure is not quite equivalent to the sampling procedure from Theorem \ref{css_intro} because we have some positive probability of choosing a $\bv{0}$ column (in fact, since $\sum_{j=1}^n \tilde \tau_j \leq 16k$, by our choice of $t$ that probability is greater than $\frac{1}{2}$ for each column). However,  Algorithm \ref{algo:css} samples from a distribution that \emph{is equivalent} to sampling from $\bv{A}$ with an all zeros column $\bv{0}$ tacked on and assigned a high ridge leverage score overestimate. Furthermore, by inspecting Algorithm \ref{algo:css}, we can see that each column is sampled \emph{independently}, as all ridge leverage score estimates are computed using the deterministic sketch $\bv{B}$. Thus, we obtain a column subset for $\left[\bv{A}\cup \bv{0}\right]$, which is clearly also a column subset for $\bv{A}$.

So, we just need to argue that we obtain an output according to the claimed distribution. Consider the state of the algorithm after each set of $t$ ``\textcolor{blue}{Process column stream}'' iterations, or equivalently, after each time the ``\textcolor{blue}{Prune columns}'' else statement is entered. Denote $\bv{A}$'s first $t$ columns as $\bv{A}^{(1)}$, its first $2t$ columns as $\bv{A}^{(2)}$, and in general, its first $m\cdot t$ columns as $\bv{A}^{(m)}$. These submatrices represent the columns of $\bv{A}$ processed by the end of each epoch of $t$ ``\textcolor{blue}{Process column stream}'' iterations.
Let's take as an inductive assumption that after every prior set of $t$ steps, each column in $\bv{C}$ equals:

\begin{align}
\label{css_probability_cond}
\bv{c} = \begin{cases}
\bv{a}_j \in [\bv{A}^{(m)}] \text{ with probability } p_j \in \left[\frac{1}{2}\frac{\tilde \tau_j c(k\log k + k\log(1/\delta)/\epsilon)}{t}, \frac{\tilde \tau_j c(k\log k + k\log(1/\delta)/\epsilon)}{t}\right],\\
\bv{0} \text{ with probability } (1- \sum_j p_j),
\end{cases}
\end{align}
where $\tilde \tau_j \ge 2\bar\tau_j(\bv{A}^{(m)})$ for all $j$ and $\sum_{j} \tilde \tau_j \leq 16k$. This is simply equivalent to our claimed output property of $\bv{C}$ once all columns have been processed.

\eqref{css_probability_cond} holds for $\bv{A}^{(1)}$ because all of its columns are initially stored in the buffer $\bv{D}$ and each $\bv{c}$ is set to $\bv{d}_j$ with probability $p_j = \frac{\tilde \tau_j c(k\log k + k\log(1/\delta)/\epsilon)}{t}$ (see line 19). From Lemma \ref{freq_dirs:lev_approx} and our chosen scaling by $4$ (line 3 of ApproximateRidgeScores), we know that $\tilde \tau_j \geq 2\bar{\tau}_j(\bv{A}^{(1)})$. Additionally, $\tilde \tau_j \leq 8\bar{\tau}_j(\bv{A}^{(1)})$, so it follows from Lemma \ref{sumScores} that $\sum_{j} \tilde \tau_j \leq 16k$.

For future iterations, $\bv{A}^{(m)}$ equals $[\bv{A}^{(m-1)}, \bv{D}$]. Consider the columns in $\bv{A}^{(m-1)}$ first. By our inductive assumption each column in $\bv{C}$ has already been set to $\bv{a}_j \in \bv{A}^{(m-1)}$ with probability $p_j \in \left[\frac{1}{2}\frac{\tilde \tau^{old}_j c(k\log k + k\log(1/\delta)/\epsilon)}{t}, \frac{\tilde \tau^{old}_j c(k\log k + k\log(1/\delta)/\epsilon)}{t}\right]$. Our ``\textcolor{blue}{Rejection sample}'' step additionally filters out any column sampled with probability $\tilde \tau_{j}/\tilde \tau^{old}_{j}$, meaning that in total $\bv{a}_j$ is sampled with the desired probability from \eqref{css_probability_cond}. We note that $\tilde \tau_{j}/\tilde \tau^{old}_{j}$ is trivially $\leq 1$ since $\tilde \tau_{j}$ was set to the minimum of $\tilde \tau^{old}_{j}$ and the ridge leverage score of $\bv{a}_j$ computed with respect to our updated Frequent Directions sketch (see line 10).

If it was set based on the updated Frequent Directions sketch, then the argument that $\tilde \tau_j \ge 2\bar\tau_j(\bv{A}^{(m)})$ is the same as for $\bv{A}^{(1)}$. On the other hand, if $\tilde \tau_{j}$ was set to equal $\tilde \tau^{old}_{j}$, then we can apply Lemma \ref{monotonicity}: from the inductive assumption, $\tilde \tau_{j} = \tilde \tau^{old}_{j} \geq 2\bar \tau_{j}(\bv{A}^{(m-1)})$ and $\bar\tau_{j}(\bv{A}^{(m-1)}) \geq \bar\tau_{j}(\bv{A}^{(m)})$ from the monotonicity property so $\tilde \tau_j \ge 2\bar\tau_j(\bv{A}^{(m)})$.

Next consider any $\bv{a}_j \in \bv{D}$. Each column $\bv{c}$ is set to $\bv{a}_j$ with the correct probability for \eqref{css_probability_cond}, but only \emph{conditioned on the fact} that $\bv{c} = \bv{0}$ before the ``\textcolor{blue}{Sample from new columns in D}'' if statement is reached. This conditioning should mean that we effectively sample each $\bv{a}_j \in \bv{D}$ with lower probability. However, the probability cannot be much lower: by our choice of $t$ and the inductive assumption on $\sum_{j} \tilde \tau_j$,  every column is set to $\bv{0}$ with \emph{at minimum} $1/2$ probability. Accordingly, $\bv{c}$ is available at least half the time, meaning that we at least sample $\bv{a}_j$ with probability $p_j = \frac{1}{2}\frac{\tilde \tau_j c(k\log k + k\log(1/\delta)/\epsilon)}{t}$, which satisfies \eqref{css_probability_cond}.

All that is left to argue is that $\sum_{j} \tilde \tau_j \leq 16k$ for $\bv{A}^{(m)}$. The argument is the same as for $\bv{A}^{(1)}$, the only difference being that for some values of $j$, we could have set $\tilde \tau_{j} = \tilde \tau^{old}_{j}$, which only decreases the total sum. 
We conclude by induction that \eqref{css_probability_cond} holds for $\bv{A}$ itself, and thus $\bv{C}$ is a $(1+\epsilon)$ error column subset (Theorem \ref{css_intro}). Algorithm \ref{algo:css} requires $O(nk)$ space to store $\bv{B}$ and maintains at most $t = O(k\log k + k\log(1/\delta)/\epsilon)$ sampled columns.  It thus proves Theorem \ref{thm:streaming_css}:
\begin{theorem}[Streaming Column Subset Selection]\label{thm:streaming_css} There exists a streaming algorithm that uses just a single-pass over $\bv{A}$'s columns to compute a $(1+\epsilon)$ error column subset $\bv{C}$ with $O(k\log k + k\log(1/\delta)/\epsilon)$ columns. The algorithm uses $O(nk)$ space in addition to the space required to store $\bv{C}$ and succeeds with probability $1-\delta$.
\end{theorem}

We note that, by using the stronger containment condition of Theorem \ref{css_intro} and the streaming projection-cost preserving sampling algorithm described below we can easily modify the above algorithm to output an optimally sized column subset with $O(k/\epsilon)$ columns. In order to select this subset, we require a Frequent Directions sketch with $\epsilon$ error, so that we can evaluate each $O(k/\epsilon)$ sized subset in our set of $O(k\log(1/\delta)/\epsilon)$ `adaptively sampled' columns and return one giving $\epsilon$ error. The higher accuracy Frequent Directions sketch incurs space overhead $O(nk/\epsilon)$.

\subsection{Streaming Projection-Cost Preserving Samples}
Our single-pass streaming procedure for  projection-cost preserving samples is similar to Algorithm \ref{algo:css}, although with one important difference. When constructing column subsets, we sampled new columns in the buffer $\bv{D}$ while ignoring the fact that ``available slots'' in $\bv{C}$ (i.e. columns currently set to $\bv{0}$) had already been consumed with some probability. This decision was deliberate, rather than a convenience for analysis. We could not account for the probability of slots being unavailable because calculating that probability precisely would require knowing the ridge leverage scores of already discarded columns.

Fortunately, the probability of a column not being set to $\bv{0}$ was bounded by $1/2$ and our procedure hits its sampling target up to this factor. However, while a constant factor approximation to sampling probabilities is also sufficient for our Theorem \ref{pcp_intro} projection-cost preservation result, the fact that columns need to be reweighted by the inverse of their sampling probability adds a complication: we do not know the \emph{true} probability with which we sampled each column! 

Unfortunately, approximating the reweighting up to a constant factor is insufficient. We need to reweight columns by a factor within $\sqrt{(1\pm \epsilon)}$ of $1/\sqrt{tp_i}$ for Theorem \ref{matrix_chernoff} and Lemma \ref{trace:bound} to hold (which are both required for Theorem \ref{pcp_intro}). This is easily checked by noting that such a reweighting is equivalent to replacing  $\bv{C}\bv{C}^T$ with $\bv{C}\bv{W}\bv{C}^T$ where $(1-\epsilon)\bv{I}_{d\times d} \preceq \bv{W} \preceq(1+\epsilon)\bv{I}_{d\times d}$.

We achieve this accuracy by modifying our algorithm so that it maintains an even higher ``open rate'' within $\bv{C}$. Specifically, we choose $t$ so that each column $\bv{c}$ has at least a $(1-\epsilon)$ probability of equaling $\bv{0}$ at any given point in our stream. The procedure is given as Algorithm \ref{algo:pcp}. The constant $c$ is the required oversampling parameter from Theorem \ref{css_intro}.

\begin{breakablealgorithm}
\caption{\algoname{Streaming Projection-Cost Preserving Samples}}
{\bf input}: $\bv{A} \in \mathbb{R}^{n \times d}$, accuracy $\epsilon$, success probability $(1-\delta)$ \\
{\bf output}: $\bv{C} \in \mathbb{R}^{n \times t}$ such that $t = \frac{1}{16\epsilon}ck\log(k/\delta)/\epsilon^2$ and each column $\bv{c}_i$ is equal to column $\frac{1}{\sqrt{\tilde \tau_j ck\log(k/\delta)/\epsilon^2}}\bv{a}_j$ with probability $p_j \in \left[(1-\epsilon)\frac{\tilde \tau_j ck\log(k/\delta)/\epsilon^2}{t}, \frac{\tilde \tau_j ck\log(k/\delta)/\epsilon^2}{t}\right]$ and $\bv{0}$ otherwise, where $\tilde \tau_j \ge 2\bar \tau_j(\bv{A})$ for all $j$ and $\sum_{j=1}^n \tilde \tau_j \leq 16k$. 
\begin{algorithmic}[1]\label{streaming_algo_pcp}
\State $count := 1$, $\bv{C} := \bv{0}_{n \times t}$, $\bv{D}:= \bv{0}_{n \times t}$, $frobA:=0$ \Comment{\textcolor{blue}{Initialize storage}}
\State $[\tilde \tau^{old}_{1},...,\tilde \tau^{old}_t] := 1$
\Comment{\textcolor{blue}{Initialize sampling probabilities}}
\For{$i := 1,\ldots, d$} \Comment{\textcolor{blue}{Process column stream}}
\State $\bv{B} := \fd(\bv{B},\bv{a}_i)$
\If{$count \le t$} \Comment{\textcolor{blue}{Collect $t$ new columns}}
\State $\bv{d}_{count} := \bv{a}_i$.
\State $frobA := frobA +\norm{\bv{a}_i}_2^2$ \Comment{\textcolor{blue}{Update $\norm{\bv{A}}_F^2$}}
\State{$count := count + 1$}
\Else \Comment{\textcolor{blue}{Prune columns}}
\State $[\tilde \tau_{1},...,\tilde \tau_t] := \min \left \{[\tilde \tau^{old}_{1},...,\tilde \tau^{old}_t],\re(\bv{B},\bv{C}, frobA)\right \}$
\State $[\tilde \tau^\bv{D}_{1},...,\tilde \tau^\bv{D}_{t}] := \re(\bv{B},\bv{D}, frobA)$
\For {$j := 1,\ldots, t$} 
\If{$\bv{c}_j \neq \bv{0}$} \Comment{\textcolor{blue}{Rejection sample}}
\State With probability $\left(1 - \tilde \tau_{j}/\tilde \tau^{old}_{j}\right)$ set $\bv{c}_j := \bv{0}$ and set $\tilde \tau^{old}_{j} := 1$. 
\State Otherwise set $\tilde \tau^{old}_{j} := \tilde \tau_{j}$ and multiply $\bv{c}_j$ by $\sqrt{\tilde \tau^{old}_{j}/\tilde \tau_{j}}$.
\EndIf
\If{$\bv{c}_j = \bv{0}$} \Comment{\textcolor{blue}{Sample from new columns in $\bv{D}$}}
\For{$\ell := 1,\ldots, t$}
\State With probability $\frac{\tilde \tau_\ell ck\log(k/\delta)/\epsilon^2}{t}$ set $\bv{c}_j := \frac{1}{\sqrt{\tilde \tau_\ell ck\log(k/\delta)/\epsilon^2}}\bv{d}_{\ell}$ and set $\tilde \tau^{old}_{j} := \tilde \tau_{\ell}$
\EndFor
\EndIf
\EndFor
\State $count := 0$
\EndIf
\EndFor
\end{algorithmic}
\begin{algorithmic}[1]
\Function{$\re$}{$\bv{B}$, $\bv{M}\in \R^{n\times t}$, $frobA$} 
\For {$i := t+1,\ldots, d$}
\State $\tilde \tau_k$:= $4\bv{m}_i^T\left(\bv{B}\bv{B}^T + \frac{frobA - \|\bv{B}_k\|_F^2}{k}\bv{I}\right)^+\bv{m}_i$
\EndFor
\State \textbf{return} $[\tilde \tau_{1},...,\tilde \tau_t]$
\EndFunction
\end{algorithmic}
\label{algo:pcp}
\end{breakablealgorithm}

The analysis of Algorithm \ref{algo:pcp} is equivalent to that of Algorithm \ref{algo:css}, along with the additional observation that our true sampling probability, $p_j$, is within an $\epsilon$ factor of the sampling probability used for reweighting, $\frac{\tilde \tau_j ck\log(k/\delta)/\epsilon^2}{t}$. Note that while $\bv{C}$ contains just $O(k\log(k/\delta)/\epsilon^2)$ non-zero columns in expectation, during the course of a the column stream it could contain as many as $O(k\log(k/\delta)/\epsilon^3)$ columns. Regardless, it is always possible to resample from $\bv{C}$ after running Algorithm \ref{algo:pcp} to construct an optimally sized sample for $\bv{A}$ with error $(1+2\epsilon)$. Overall we have:
\begin{theorem}[Streaming Projection-Cost Preserving Sampling]\label{thm:streaming_pcp} There exists a streaming algorithm that uses just a single-pass over $\bv{A}$'s columns to compute a $(1+\epsilon)$ error projection-cost preserving sample $\bv{C}$ with $O(k\log (k/\delta)/\epsilon^2)$ columns. The algorithm requires a fixed $O(nk)$ space overhead along with space to store $O(k\log (k/\delta)/\epsilon^3)$ columns of $\bv{A}$. It succeeds with probability $1-\delta$.
\end{theorem}

\bibliography{ridgeLeverageScores}{}
\bibliographystyle{alpha}

\appendix
\section{Trace Bound for Ridge Leverage Score Sampling} 
\label{app:trace_bound}
\begin{lemma}
\label{trace:bound}
For $i \in \{1,\ldots,d\}$, let $\tilde \tau_i \ge \bar \tau_i(\bv{A})$ be an overestimate for the $i^{th}$ ridge leverage score. Let $p_i = \frac{\tilde \tau_i}{\sum_i \tilde \tau_i}$. Let $t = \frac{c\log(k/\delta)}{\epsilon^2} \cdot \sum_i \tilde \tau_i$, for some sufficiently large constant $c$. Construct $\bv{C}$ by sampling $t$ columns of $\bv{A}$, each set to $\frac{1}{\sqrt{tp_i}}\bv{a}_i$ with probability $p_i$.  Let $m$ be the index of the smallest singular value with $\sigma_m^2 \ge \frac{\norm{\bv{A}_{\setminus k}}_F^2}{k}$. With probability $1-\delta$, $\bv{C}$ satisfies:
\begin{align}\label{trace_guarantee}
|\tr(\bv{A}_{\setminus m}\bv{A}_{\setminus m}^T) -  \tr(\bv{U}_{\setminus m}\bv{U}_{\setminus m}^T\bv{C}\bv{C}^T\bv{U}_{\setminus m}\bv{U}_{\setminus m}^T)| \leq \epsilon \norm{\bv{A}_{\setminus k}}_F^2.
\end{align}
\end{lemma}
\begin{proof}
Letting $\bv{P}_{\setminus m} = \bv{U}_{\setminus m}\bv{U}_{\setminus m}^T$, we can rewrite \eqref{trace_guarantee} as:
\begin{align*}
|\norm{\bv{P}_{\setminus m} \bv{C}}_F^2 - \norm{\bv{A}_{\setminus m}}_F^2| \le \epsilon \norm{\bv{A}_{\setminus k}}_F^2.
\end{align*}
We can write $\norm{\bv{P}_{\setminus m} \bv{C}}_F^2$ as a sum over column norms:
\begin{align*}
\norm{\bv{P}_{\setminus m} \bv{C}}_F^2 = \sum_{j=1}^t \norm{\bv{P}_{\setminus m}\bv{c}_j}_2^2.
\end{align*}
Now, for some $i \in \{1,\ldots,d\}$ and recalling our definition $\bv{\bar \Sigma}^2_{i,i} = \sigma^2_i(\bv{A}) + \frac{\norm{\bv{A}_{\setminus k}}_F^2}{k}$ we have:
\begin{align*}
\norm{\bv{P}_{\setminus m}\bv{c}_i}_2^2 = \frac{1}{tp_i} \norm{\bv{P}_{\setminus m}\bv{a}_i}_2^2 &\le
\frac{\epsilon^2}{c\log(k/\delta)} \cdot\frac{\norm{\bv{P}_{\setminus m}\bv{a}_i}_2^2}{\bar \tau_i(\bv{A})}\\
&= \frac{\epsilon^2}{c\log(k/\delta)} \cdot \frac{\norm{\bv{P}_{\setminus m}\bv{a}_i}_2^2}{\bv{a}_i^T\left(\bv{U}\bv{\bar \Sigma}^{-2}\bv{U}^T\right)\bv{a}_i}\\ &\le\frac{\epsilon^2}{c\log(k/\delta)} \cdot \frac{\norm{\bv{P}_{\setminus m}\bv{a}_i}_2^2}{\bv{a}_i^T\bv{P}_{\setminus m}\left(\bv{U}\bv{\bar \Sigma}^{-2}\bv{U}^T\right)\bv{P}_{\setminus m}\bv{a}_i} \\
&= \frac{\epsilon^2}{c\log(k/\delta)} \cdot \frac{\norm{\bv{P}_{\setminus m}\bv{a}_i}_2^2}{\bv{a}_i^T\bv{P}_{\setminus m}\left(\bv{U}_{\setminus m}\bv{\bar \Sigma}^{-2}\bv{U}_{\setminus m}^T\right)\bv{P}_{\setminus m}\bv{a}_i} \\
&\le \frac{\epsilon^2}{c\log(k/\delta)} \cdot \frac{\norm{\bv{P}_{\setminus m}\bv{a}_i}_2^2}{\frac{k}{2\norm{\bv{A}_{\setminus k}}_F^2}\norm{\bv{P}_{\setminus m}\bv{a}_i}_2^2}\\
& \le \frac{2\epsilon^2}{c\log(k/\delta)} \cdot\norm{\bv{A}_{\setminus k}}_F^2,
\end{align*}
where the second to last inequality follows from the fact that 
$\bv{\bar \Sigma}^2_{i,i} = \sigma^2_i(\bv{A}) + \frac{\norm{\bv{A}_{\setminus k}}_F^2}{k} \le \frac{2\norm{\bv{A}_{\setminus k}}_F^2}{k}$ for $i \ge m$. Therefore, $\bv{U}_{\setminus m}\bv{\bar \Sigma}^{-2}\bv{U}_{\setminus m}^T \succeq \frac{k}{2\norm{\bv{A}_{\setminus k}}_F^2} \bv{P}_{\setminus m}$.

So, $\frac{c\log(k/\delta)}{2\epsilon^2\norm{\bv{A}_{\setminus k}}_F^2} \cdot \norm{\bv{P}_{\setminus m}\bv{c}_i}_2^2 \in [0,1]$. We have $\E \left [ \sum_{j=1}^t \norm{\bv{P}_{\setminus m}\bv{c}_i}_2^2 \right ] = \norm{\bv{A}_{\setminus m}}_F^2$ so by a Chernoff bound:
\begin{align*}
&\Pr \left [ \norm{\bv{P}_{\setminus m}\bv C}_F^2 \ge \norm{\bv{A}_{\setminus m}}_F^2 + \epsilon \norm{\bv{A}_{\setminus k}}_F^2\right ] \\&= \Pr \left [ \frac{c\log(k/\delta)}{2\epsilon^2\norm{\bv A_{\setminus k}}_F^2}\sum_{j=1}^t \norm{\bv{P}_{\setminus m}\bv{c}_i}_2^2 \ge  \left (1 + \frac{\epsilon \norm{\bv{A}_{\setminus k}}_F^2}{ \norm{\bv{A}_{\setminus m}}_F^2} \right )\frac{c\log(k/\delta)\norm{\bv{A}_{\setminus m}}_F^2}{2\epsilon^2\norm{\bv A_{\setminus k}}_F^2} \right ]\\
&\le e^{-c\log(k/\delta)/4} \le \delta/2,
\end{align*} 
if we set $c$ sufficiently large. In the second to last step we use the fact that $\frac{ \norm{\bv{A}_{\setminus k}}_F^2}{ \norm{\bv{A}_{\setminus m}}_F^2} \ge \frac{1}{2}$ by the definition of $m$. We can similarly prove that $\Pr \left [ \norm{\bv{P}_{\setminus m}\bv C}_F^2 \leq \norm{\bv{A}_{\setminus m}}_F^2 - \epsilon \norm{\bv{A}_{\setminus k}}_F^2\right ] \leq \delta/2$. Union bounding gives the result.
\end{proof}

\section{Independent Sampling Bounds}
\label{app:without_replacement}
In this section we give analogies to Theorem \ref{matrix_chernoff} and Lemma \ref{trace:bound} when columns are sampled independently using their ridge leverage scores rather than sampled with replacement.

\begin{lemma}\label{matrix_chernoff2}
For $i \in \{1,\ldots,d\}$, given $\tilde \tau_i \ge \bar \tau_i(\bv{A})$ for all $i$, let $p_i = \min \left \{\tilde \tau_i \cdot \frac{c\log(k/\delta)}{\epsilon^2}, 1\right \}$ for some sufficiently large constant $c$. Construct $\bv{C}$ by independently sampling each column $\bv{a}_i$ from $\bv{A}$ with probability $p_i$ and scaling selected columns by $1/\sqrt{p_i}$. With probability $1-\delta$, $\bv{C}$ has $O\left (\log(k/\delta)/\epsilon^2 \cdot \sum_i \tilde \tau_i\right )$ columns and satisfies:
\begin{align}
(1-\epsilon) \bv{C} \bv{C}^T - \frac{\epsilon}{k} \norm{\bv{A} - \bv{A}_k}_F^2 \bv{I}_{n\times n} \preceq \bv{A}\bv{A}^T \preceq (1+\epsilon) \bv{C} \bv{C}^T + \frac{\epsilon}{k} \norm{\bv{A} - \bv{A}_k}_F^2 \bv{I}_{n\times n}\tag{\ref{freq_dirs_two_sided}}.
\end{align}
\end{lemma}
\begin{proof}
Again we rewrite the ridge leverage score definition
using $\bv{A}$'s singular value decomposition:
\begin{align*}
\bar \tau_i(\bv{A}) &= \bv{a}_i^T \left (\bv{U} \bv{\Sigma}^2 \bv{U}^T + \frac{\norm{\bv{A}_{\setminus k}}_F^2}{k} \bv{U}  \bv{U}^T \right)^+ \bv{a}_i\\
&= \bv{a}_i^T \left (\bv{U} \bv{\bar \Sigma}^2 \bv{U}^T \right )^+ \bv{a}_i
=  \bv{a}_i^T \left (\bv{U} \bv{\bar \Sigma}^{-2} \bv{U}^T \right ) \bv{a}_i,
\end{align*}
where $\bv{\bar \Sigma}^2_{i,i} = \sigma^2_i(\bv{A}) + \frac{\norm{\bv{A}_{\setminus k}}_F^2}{k}$.
For each $i \in 1,\ldots,d$ define the matrix valued random variable:
\begin{align*}
\bv{X}_i = 
\begin{cases}
\left (\frac{1}{p_i} -1 \right ) \bv{\bar \Sigma}^{-1} \bv{U}^T\bv{a}_i \bv{a}_i^T \bv{U} \bv{\bar \Sigma}^{-1} \text{ with probability } p_i\\
-\bv{\bar \Sigma}^{-1} \bv{U}^T\bv{a}_i \bv{a}_i^T \bv{U} \bv{\bar \Sigma}^{-1} \text{ with probability } (1-p_i)
\end{cases}
\end{align*}

Let $\bv{Y} = \sum \bv{X}_i$. We have $\E \bv{Y} = \bv{0}$. Furthermore, $\bv{CC}^T = \bv{U}\bv{\bar \Sigma} \bv{Y} \bv{\bar \Sigma} \bv{U} + \bv{AA}^T$. Showing $\norm{\bv{Y}}_2 \le \epsilon$ gives $-\epsilon \bv{I} \preceq \bv{Y} \preceq  \epsilon \bv{I}$, and since $\bv{U}\bv{\bar \Sigma}^2 \bv{U}^T = \bv{AA}^T +\frac{ \norm{\bv{A}_{\setminus k}}_F^2}{k} \bv{I}$ would give:
\begin{align*}
(1-\epsilon)\bv{AA}^T - \frac{\epsilon \norm{\bv{A}_{\setminus k}}_F^2}{k} \bv{I} \preceq \bv{C}\bv{C}^T \preceq (1+\epsilon)\bv{AA}^T + \frac{\epsilon \norm{\bv{A}_{\setminus k}}_F^2}{k} \bv{I}.
\end{align*}
After rearranging and adjusting constants on $\epsilon$, this statement is equivalent to \eqref{freq_dirs_two_sided}.

To prove that $\norm{\bv{Y}}_2$ is we use the same stable rank matrix Bernstein inequality used for our with replacement results \cite{tropp2015introduction}.
If $p_i = 1$ (i.e. $\tilde  \tau_i \cdot c\log(k/\delta)/\epsilon^2 \ge 1$) then $\bv{X}_i = \bv{0}$ so $\norm{\bv{X}_i}_2 = 0$. Otherwise, we use the fact that $\frac{1}{\bar \tau_i(\bv{A})} \bv{a}_i \bv{a}_i^T \preceq \bv{A}\bv{A}^T + \frac{ \norm{\bv{A}_{\setminus k}}_F^2}{k} \bv{I}$, which lets us bound:
\begin{align*}
\frac{1}{\tilde \tau_i} \cdot \bv{\bar \Sigma}^{-1} \bv{U}^T\bv{a}_i \bv{a}_i^T \bv{U} \bv{\bar \Sigma}^{-1} \preceq \bv{\bar \Sigma}^{-1} \bv{U}^T\left(\bv{A}\bv{A}^T + \frac{ \norm{\bv{A}_{\setminus k}}_F^2}{k} \bv{I}\right) \bv{U} \bv{\bar \Sigma}^{-1} = \bv{I}.
\end{align*}
 So we have $\bv{X}_i \preceq \frac{1}{p_i} \bv{\bar \Sigma}^{-1} \bv{U}^T\bv{a}_i \bv{a}_i^T \bv{U} \bv{\bar \Sigma}^{-1} \preceq\frac{\epsilon^2}{c\log(k/\delta)}\bv{I}$ and hence $
\norm{\bv{X}_i}_2 \le \frac{\epsilon^2}{c\log(k/\delta)}.$

Next we bound the variance of $\bv{Y}$. 
\begin{align}
\label{eq:var_bound2}
\E (\bv{Y}^2) = \sum \E (\bv{X}_i^2 ) &\preceq \sum \left [p_i \left (\frac{1}{p_i}-1\right )^2 + (1-p_i) \right ] \cdot \bv{\bar \Sigma}^{-1} \bv{U}^T\bv{a}_i \bv{a}_i^T \bv{U} \bv{\bar \Sigma}^{-2} \bv{U}^T\bv{a}_i \bv{a}_i^T \bv{U} \bv{\bar \Sigma}^{-1}\nonumber\\
&\preceq \sum \frac{1}{p_i} \cdot  \bar \tau_i(\bv{A}) \cdot \bv{\bar \Sigma}^{-1} \bv{U}^T\bv{a}_i \bv{a}_i^T \bv{U} \bv{\bar \Sigma}^{-1}
\preceq \frac{\epsilon^2}{c\log(k/\delta)} \bv{\bar \Sigma}^{-1} \bv{U}^T\bv{A}\bv{A}^T \bv{U} \bv{\bar \Sigma}^{-1}\nonumber\\
&\preceq  \frac{\epsilon^2}{c\log(k/\delta)} \bv \Sigma^2 \bv{\bar \Sigma}^{-2} 
\preceq \frac{\epsilon^2}{c\log(k/\delta)} \bv{D}.
\end{align}
where again we set $\bv{D}_{i,i} = 1$ for $i \in 1,\ldots,k$ and $\bv{D}_{i,i} = \frac{\sigma_i^2}{\sigma_i^2 + \norm{\bv{A}_{\setminus k}}_F^2/k}$ for all $i \in k+1,\ldots,n$. By the stable rank matrix Bernstein inequality given in Theorem 7.3.1 of \cite{tropp2015introduction}, for $\epsilon < 1$,
\begin{align}
\label{eq:almost_chernoff2}
\Pr \left [\norm{\bv Y} \ge \epsilon \right ] &\le \frac{4\tr(\bv D) }{\norm{\bv{D}}_2} e^{\frac{-\epsilon^2/2}{\left (\frac{\epsilon^2}{c\log(k/\delta)}(\norm{\bv{D}}_2+\epsilon/3)\right )}}.
\end{align}
Clearly $\norm{\bv D}_2 = 1$. Furthermore, 
\begin{align*}
\tr(\bv{D}) = k + \sum_{i=k+1}^d\frac{\sigma_i^2}{\sigma_i^2 + \frac{\norm{\bv{A}_{\setminus k}}_F^2}{k}} \leq k +\sum_{i=k+1}^d\frac{\sigma_i^2}{\frac{\norm{\bv{A}_{\setminus k}}_F^2}{k}} = k +\frac{\sum_{i=k+1}^d\sigma_i^2}{\frac{\norm{\bv{A}_{\setminus k}}_F^2}{k}} \leq k+ k.
\end{align*}
Plugging into \eqref{eq:almost_chernoff}, we see that
\begin{align*}
\Pr \left [\norm{\bv Y} \ge \epsilon \right ] &\le 8k e^{- \frac{c\log(k/\delta)}{2}})\le \delta/2,
\end{align*}
if we choose the constant $c$ large enough. 
So we have established \eqref{freq_dirs_two_sided}. 

All that remains to note is that, the expected number of columns in $\bv{C}$ is at most  $\frac{c\log(k/\delta)}{\epsilon^2} \cdot \sum_{i=1}^d \tilde \tau_i$. Accordingly, $\bv{C}$ has at most $O\left (\frac{\log(k/\delta)}{\epsilon^2} \cdot \sum_i \tilde \tau_i \right)$ columns with probability $> 1-\delta/2$ by a standard Chernoff bound. Union bounding over failure probabilities gives the lemma.
\end{proof}

\begin{lemma}\label{trace:bound2}
For $i \in \{1,\ldots,d\}$, given $\tilde \tau_i \ge \bar \tau_i(\bv{A})$ for all $i$, let $p_i = \min \left \{\bar \tau_i(\bv A) \cdot \frac{c\log(k/\delta)}{\epsilon^2}, 1\right \}$ for some sufficiently large constant $c$. Construct $\bv{C}$ by independently sampling each column $\bv{a}_i$ from $\bv{A}$ with probability $p_i$ and scaling selected columns by $1/\sqrt{p_i}$.  Let $m$ be the index of the smallest singular value with $\sigma_m^2 \ge \frac{\norm{\bv{A}_{\setminus k}}_F^2}{k}$. With probability $1-\delta$, $\bv{C}$ satisfies:
\begin{align}\label{trace_guarantee2}
|\tr(\bv{A}_{\setminus m}\bv{A}_{\setminus m}^T) -  \tr(\bv{U}_{\setminus m}\bv{U}_{\setminus m}^T\bv{C}\bv{C}^T\bv{U}_{\setminus m}\bv{U}_{\setminus m}^T)| \leq \epsilon \norm{\bv{A}_{\setminus k}}_F^2.
\end{align}
\end{lemma}
\begin{proof}

We need to show $\tr(\bv{A}_{\setminus m}\bv{A}_{\setminus m}^T) -  \tr(\bv{U}_{\setminus m}\bv{U}_{\setminus m}^T\bv{B}\bv{B}^T\bv{U}_{\setminus m}\bv{U}_{\setminus m}^T) \ge -\epsilon \norm{\bv{A}_{\setminus m}}_F^2.$ Letting $\bv{P}_{\setminus m} = \bv{U}_{\setminus m}\bv{U}_{\setminus m}^T$, we can rewrite this as:
\begin{align*}
\norm{\bv{P}_{\setminus m} \bv{B}}_F^2 - \norm{\bv{A}_{\setminus m}}_F^2 \le \epsilon \norm{\bv{A}_{\setminus m}}_F^2.
\end{align*}
We can write $\norm{\bv{P}_{\setminus m} \bv{B}}_F^2$ as a sum over column norms:
\begin{align*}
\norm{\bv{P}_{\setminus m} \bv{B}}_F^2 = \sum_{i=1}^d \mathcal{I}_i \frac{1}{p_i} \norm{\bv{P}_{\setminus m}\bv{a}_i}_2^2,
\end{align*}
where $ \mathcal{I}_i$ is an indicator random variable equal to $1$ with probability $p_i$ and $0$ otherwise. 

We have:
\begin{align*}
\frac{1}{p_i} \norm{\bv{P}_{\setminus m}\bv{a}_i}_2^2 = \frac{\epsilon^2}{c\log(k/\delta)} \cdot\frac{\norm{\bv{P}_{\setminus m}\bv{a}_i}_2^2}{ \tilde  \tau_i} &\le \frac{\epsilon^2}{c\log(k/\delta)} \cdot \frac{\norm{\bv{P}_{\setminus m}\bv{a}_i}_2^2}{\bv{a}_i^T\left(\bv{U}\bv{\bar \Sigma}^{-2}\bv{U}^T\right)\bv{a}_i}\\ &\le\frac{\epsilon^2}{c\log(k/\delta)} \cdot \frac{\norm{\bv{P}_{\setminus m}\bv{a}_i}_2^2}{\bv{a}_i^T\bv{P}_{\setminus m}\left(\bv{U}\bv{\bar \Sigma}^{-2}\bv{U}^T\right)\bv{P}_{\setminus m}\bv{a}_i} \\
&= \frac{\epsilon^2}{c\log(k/\delta)} \cdot \frac{\norm{\bv{P}_{\setminus m}\bv{a}_i}_2^2}{\bv{a}_i^T\bv{P}_{\setminus m}\left(\bv{U}_{\setminus m}\bv{\bar \Sigma}^{-2}\bv{U}_{\setminus m}^T\right)\bv{P}_{\setminus m}\bv{a}_i} \\
&\le \frac{\epsilon^2}{c\log(k/\delta)} \cdot \frac{\norm{\bv{P}_{\setminus m}\bv{a}_i}_2^2}{\frac{k}{2\norm{\bv{A}_{\setminus k}}_F^2}\norm{\bv{P}_{\setminus m}\bv{a}_i}_2^2}\\
& \le \frac{2\epsilon^2}{c\log(k/\delta)} \cdot\norm{\bv{A}_{\setminus k}}_F^2,
\end{align*}
where the second to last inequality follows from the fact that 
$\bv{\bar \Sigma}^2_{i,i} = \sigma^2_i(\bv{A}) + \frac{\norm{\bv{A}_{\setminus k}}_F^2}{k} \le \frac{2\norm{\bv{A}_{\setminus k}}_F^2}{k}$ for $i \ge m$. Therefore, $\bv{U}_{\setminus m}\bv{\bar \Sigma}^{-2}\bv{U}_{\setminus m}^T \succeq \frac{k}{2\norm{\bv{A}_{\setminus k}}_F^2} \bv{P}_{\setminus m}$.

So $\frac{c\log(k/\delta)}{2\epsilon^2\norm{\bv{A}_{\setminus m}}_F^2} \cdot \frac{1}{p_i} \norm{\bv{P}_{\setminus m}\bv{a}_i}_2^2 \in [0,1]$ and by a Chernoff bound we have:
\begin{align*}
\Pr \left [ \norm{\bv{P}_{\setminus m}\bv B}_F^2 \ge (1+\epsilon) \norm{\bv{A}_{\setminus m}}_F^2 \right ]
&= \Pr \left [ \frac{c\log(k/\delta)}{2\epsilon^2\norm{\bv A_{\setminus m}}_F^2}\sum_{i=1}^d I_i \frac{1}{p_i} \norm{\bv{P}_{\setminus m}\bv{a}_i}_2^2 \ge (1+\epsilon) \frac{c\log(k/\delta)}{2\epsilon^2} \right ]\\ &\le e^{-c\log(k/\delta)/4} \le \delta/2,
\end{align*} 
if we set $c$ sufficiently large.
\end{proof}

\end{document}

%% file: bib_macros.tex
  \usepackage{nth}
  \usepackage{intcalc}